\documentclass[12pt]{article}

\advance\voffset by -2.0cm \advance\hoffset by -1.25cm
\usepackage{amsmath}
\usepackage{amssymb}
\usepackage{amsthm}
\usepackage{amsfonts}

\numberwithin{equation}{section}
\newtheorem{theorem}{Theorem}[section]
\newtheorem{proposition}[theorem]{Proposition}
\newtheorem{corollary}[theorem]{Corollary}
\newtheorem{lemma}[theorem]{Lemma}

\theoremstyle{definition}

\newcommand{\CC}{\mathbb{C}} 
\newcommand{\RR}{\mathbb{R}} 
\newcommand{\FF}{\mathbb{F}} 
\newcommand{\ZZ}{\mathbb{Z}} 
\newcommand{\NN}{\mathbb{N}} 

\DeclareMathOperator{\im}{Im} 
\DeclareMathOperator{\diag}{diag} 
\DeclareMathOperator{\Sp}{Sp} 

\newcommand{\smsum}{\mathop{\textstyle{\sum}}\limits} 

\newcommand{\B}{\mathcal{B}} 
\newcommand{\R}{\mathcal{R}} 
\newcommand{\J}{\mathcal{J}} 

\newcommand{\II}{\mathbb{I}} 
\newcommand{\Sieg}{\mathfrak{H}} 
\newcommand{\tp}{\,{}^t\!} 
\newcommand{\be}{\begin{equation}}
\newcommand{\ee}{\end{equation}}

\newcommand{\g}{\mathfrak{g}}

\newlength{\oldcolsep}
\newcommand{\spin}[2]{\addtolength{\arraycolsep}{-3pt}%
{\Bigl[\,\begin{matrix}#1\\[-2pt]#2\end{matrix}\,\Bigr]}\addtolength{\arraycolsep}{3pt}}
\newcommand{\thspin}[2]{\addtolength{\arraycolsep}{-3pt}%
{\left[\!\begin{smallmatrix}#1\\#2\end{smallmatrix}\!\right]}\addtolength{\arraycolsep}{3pt}}

\begin{document}

\title{Getting superstring amplitudes by degenerating Riemann surfaces}
\author{Marco Matone$^1$ and Roberto Volpato$^2$}\date{}

\maketitle

\begin{center}$^1$Dipartimento di Fisica ``G. Galilei'' and Istituto
Nazionale di Fisica Nucleare \\
Universit\`a di Padova, Via Marzolo, 8 -- 35131 Padova,
Italy\end{center}

\medskip

\begin{center}{$^2$Institut f\"ur Theoretische Physik,
ETH Z\"urich} \vspace*{0.1cm} \\
{8093 Z\"urich, Switzerland}\end{center}

\bigskip

\begin{abstract} We explicitly show how
the chiral superstring amplitudes can be obtained through
factorisation of the higher genus chiral measure induced by suitable
degenerations of Riemann surfaces. This powerful tool
also allows to derive, at any genera, consistency relations
involving the amplitudes and the measure.  A key point concerns the
choice of the local coordinate at the node on degenerate Riemann
surfaces that greatly simplifies the computations. As a first
application, starting from recent ans\"atze for the chiral measure
up to genus five, we compute the chiral two-point function for
massless Neveu-Schwarz states at genus two, three and four. For
genus higher than three, these computations include some new
corrections to the conjectural formulae appeared so far in the
literature. After GSO projection, the two-point function vanishes at
genus two and three, as expected from space-time supersymmetry
arguments, but not at genus four. This suggests that the ansatz for
the superstring measure should be corrected for genus higher than
four.

\end{abstract}

\newpage

\section{Introduction}

In the last years there has been a considerable progress in
understanding and deriving explicit formulas for
multiloop superstring amplitudes. In a series of papers
\cite{D'Hoker:2001zp,D'Hoker:2004xh,D'Hoker:2004ce,D'Hoker:2005jc,D'Hoker:2005ia,D'Hoker:2007ui},
D'Hoker and Phong obtained, from first principles, an explicitly
gauge independent expression for the $2$-loop chiral superstring
measure on the moduli space of Riemann surfaces, given by
\begin{equation}\label{twoloopmeas}d\mu^{(2)}[\delta]=\Xi^{(2)}[\delta]d\mu^{(2)}_{Bos}\ ,
\end{equation}
where $\delta\in\ZZ_2^2$ is an even spin structure,
$\Xi^{(2)}[\delta]$ is a modular form of weight $8$ for a subgroup
of $\Sp(4,\ZZ)$ and $d\mu^{(g)}_{Bos}$ is the genus $g$ bosonic
string measure. Based on such a result,  they also proved the
non-renor\-ma\-lisa\-tion of the cosmological constant and of the
$n$-point functions, $n\le 3$, up to $g=2$, as expected by
space-time supersymmetry arguments \cite{Martinec:1986wa}.
Furthermore, the four-point amplitude has been computed and checked
against the constraints coming from S-duality \cite{D'Hoker:2005ht}.

Direct computations of higher loop corrections to superstring
amplitudes have been intensively investigated during the years.
In spite of such efforts, direct
computations still appear out of reach. Nevertheless, the strong
constraints coming from modular invariance and from factorisation
under degeneration limits, together with the explicit $2$-loop
expressions, can lead to reliable conjectures on such corrections.
This is the point of view adopted, for example, in
\cite{Matone:2005vm}, where the explicit expression of higher loop
contributions to the four-point function has been proposed.

In \cite{D'Hoker:2004xh,D'Hoker:2004ce}, D'Hoker and Phong
conjectured that eq.\eqref{twoloopmeas} could be extended to genus
$g>2$ for a suitable modular form $\Xi^{(g)}[\delta]$ of weight $8$.
Such a form is required to fulfill a set of constraints related to
holomorphicity, modular invariance and factorisation. In
\cite{Cacciatori:2008ay} Cacciatori, Dalla Piazza and van Geemen
(CDvG) found a solution to such constraints at $g=3$ and in
\cite{DallaPiazza:2008} the uniqueness of this solution has been
proved.

The CDvG ansatz for the $g=3$ measure has been generalised to any $g$
by Grushevsky \cite{Grushevsky:2008zm}. Salvati Manni proved in
\cite{SalvatiManni:2008qa} that such an ansatz provides a solution to
the constraints for $g=4,5$. For $g>5$ some issues arise due to
the presence of holomorphic roots of modular forms in the definition
of the chiral measure, and it is not clear whether such roots are
well defined and have the correct modular properties. In
\cite{oura-2008}, Oura, Poor, Salvati Manni and Yuen (OPSMY)
proposed an alternative construction for the chiral measure up to
$g=5$, using lattice theta series rather than theta constants, as done by Grushevsky.
They also proved that the solution to the constraints is unique up
to $g=4$. The explicit equivalence of all ans\"atze up to $g=4$ has
been shown in \cite{DuninBarkowski:2009ej}. It is still an open
question to understand whether Grushevsky and OPSMY proposals
coincide at $g=5$.

\medskip

There are several consistency conditions that the chiral superstring measure must satisfy. In particular, non-renormalisation theorems from space-time
supersymmetry  imply that the cosmological constant and the $n$-point functions, for $n\le 3$, must vanish  \cite{Martinec:1986wa}.
The vanishing of the $g$-loop correction to the cosmological constant corresponds to the condition
$$\sum_{\delta\text{ even}}\Xi^{(g)}[\delta]=0\ ,
$$ where the sum over the spin structures corresponds to the GSO projection \cite{Gliozzi:1976qd}.
This identity has been proved for the CDvG-Grushevsky
(CDvG-G) ansatz for genera $3$ \cite{Cacciatori:2008ay} and $4$
\cite{Grushevsky:2008zm}. Remarkably, for $g=4$, the cosmological constant
corresponds to a non-zero Siegel modular form of weight $8$ (the
Igusa-Schottky form), vanishing only on the locus of Jacobians of
Riemann surfaces. For $g=5$ the vanishing of the cosmological constant has to be imposed
as a further constraint on the chiral measure and it is satisfied by the OPSMY ansatz and by a slight modification of the original Grushevsky's ansatz \cite{Grushevsky:2008zp}.
It would be interesting to understand whether this further condition implies the uniqueness of the solution also in the case $g=5$.

\medskip

Consistency conditions related to non-renormalisation of the chiral
amplitudes for $n=1,2,3$ Neveu-Schwarz massless states are much more
difficult to check. As the two-loop explicit computation shows,
these amplitudes are given by a sum of several different
contributions that cannot be easily determined in terms of the
chiral measure alone. Very schematically, all such contributions can
be collected into two different terms that, following
\cite{D'Hoker:2007ui}, we call the connected ($\B_c[\delta]$) and
disconnected ($\B_d[\delta]$) part (see also \cite{Lechtenfeld:1989ke} for a relevant preliminary investigation of such contributions).
The disconnected part can be
easily expressed in terms of the chiral measure. In particular, the
disconnected part of one-point function vanishes trivially after
summing over the spin structures. For $n=2,3$, $\B_d[\delta]$ is
given, up to spin-independent factors, by the functions
\begin{align}
\hat A_2[\delta](a,b)&:= \Xi^{(g)}[\delta]S_\delta(a,b)^2\ ,\label{faketwodelta}\\
\hat A_3[\delta](a,b,c)&:= \Xi^{(g)}[\delta]S_\delta(a,b)S_\delta(b,c)S_\delta(c,a)
\ ,\label{fakethreedelta}
\end{align}
where $a,b,c$ are the insertion points and $S_\delta$ is the Szeg\"o
kernel \cite{Fay:1973}. On the other hand,
$\B_c[\delta]$ is much more complicated to compute and its precise
form is unknown for $g>2$. One possible approach to this problem is
to introduce some simplifying assumptions. In this respect, it is
useful to analyse the explicit two-loop computation of the two- and
three-point functions. In these cases, the connected and
disconnected contributions vanish separately after the GSO
projection \cite{D'Hoker:2005jc}. It is reasonable to conjecture
that a similar mechanism occurs at higher genus as well, so that
$$\sum_{\delta\text{ even}}\B_c[\delta]=0\ ,$$
and the non-renormalisation theorems would imply that also $\sum_{\delta\text{ even}}\B_d[\delta]$ vanishes, i.e.
\begin{align}\hat A_2(a,b)&:=\sum_{\delta\text{ even}}\hat A_2[\delta](a,b)=0\ ,\label{twopvanish}\\
\hat A_3(a,b,c)&:=\sum_{\delta\text{ even}}\hat A_3[\delta](a,b,c)=0\ ,\label{threepvanish}
\end{align}
for all insertion points $a,b,c$.  A strong argument for the identities
\eqref{twopvanish} and \eqref{threepvanish} to hold on the
hyperelliptic locus for any genus has been given by Morozov in
\cite{Morozov:2008xd}, whereas
Grushevsky and Salvati Manni proved \eqref{twopvanish} for genus
$3$ \cite{Grushevsky:2008qp}. However, in \cite{Matone:2008td} it has been proved that
\eqref{threepvanish} does not hold for any non-hyperelliptic Riemann
surface of genus $3$. More precisely, $\hat A_3(a,b,c)=0$ for all
$a,b,c\in C$, where $C$ is a Riemann surface of genus $3$, if and
only if $C$ is hyperelliptic. In this paper, we will also prove that \eqref{twopvanish} and \eqref{threepvanish} do not hold at genus four (see subsection \ref{s:grus}).
Apparently, these results lead to a
contradiction between the chiral measure ansatz at three loop and
non-renormalisation theorems. However, as discussed in
\cite{Matone:2008td}, it is plausible to consider this discrepancy
as the evidence that the connected part of the chiral amplitude does
not vanish in these cases.

\medskip

In this paper, we propose a different approach to the computation of
the (spin dependent part of the) chiral amplitude for two NS
massless states at $g$-loop, for $g=2,3,4$, based on the natural
factorisation properties of the chiral measure. More precisely, the
two-point function can be obtained by considering the chiral measure
at genus $g+1$ in the limit in which one of the handles of the
Riemann surface becomes infinitely long. We apply this procedure to
the OPSMY ansatz for the chiral measure and show that the two-point
function is given by \eqref{faketwodelta} plus a correction term.
For $g=2,3$ such a term vanishes after summing over the spin
structures, so that the complete two-point function vanishes as
expected by space-time supersymmetry. This represents a highly
non-trivial consistency check for the chiral measure at genus
$g+1=3,4$. On the other hand, the two-point function does not vanish
at genus $4$, which could be the signal that the OPSMY ansatz must
be corrected at $g=5$.

\medskip

The paper is organised as follows. In section two, after reviewing Grushevsky ansatz, we formulate a lemma
and proposition based on theta relations, that imply the non vanishing of the proposed two-point function at genus four. This also
easily reproduces the known results in the case of genus lower than four. Another simple consequence is that the proposed three-point amplitude
does not vanish at genus four as requested by the non-renormalisation theorem. We conclude this section by considering the OPSMY ansatz for
the superstring measure in terms of theta lattices \cite{oura-2008}.

In section three we consider the degeneration of handles of Riemann
surfaces, to provide basic relationships among measure and
amplitudes at arbitrary genera. A key point is the choice of a local coordinate at the node of the degenerate
Riemann surfaces that greatly simplifies the computations. As an application, we explicitly
show that the two-point function corresponds to the leading term in
the degeneration parameter. It turns out that the proposed
superstring measure actually leads to a vanishing two-point
function for genus two and three. In this respect, it should be
stressed that while the results in section two are obtained assuming the form \eqref{faketwodelta} and \eqref{fakethreedelta} for the $n$-point functions, here the results are
obtained using only the ansatz for
the chiral measure, so that this investigation also provides an important check for the
ansatz itself at genus three and four. We also directly show that
the two-point function at $g=4$, implied by the OPSMY ansatz for the
measure, does not vanish as requested by the non-renormalisation
theorem. In turn, this also implies that the proposed three-point
function does not vanish at the same genus. Section four is devoted
to our conclusions.

In appendix A we first fix some notation used in the main text and
recall basic facts on Riemann surfaces and Riemann theta functions.
Next, we provide a careful analysis of the degeneration of Riemann surfaces which is used in section three to derive the two-point function. We also
provide a basic formula for a section of $|2\Theta|$, with
$\Theta$ denoting the theta divisor. In appendix B, after reviewing
useful results on unimodular lattices and the associated theta
series, we consider the summation on the spin structures. In this
context, we obtain some results that, at the best of our knowledge,
are new.

\section{The chiral superstring measure}

The chiral superstring measure $d\mu^{(g)}[\delta]$ satisfies some
natural consistency conditions coming from modular invariance and
factorisation properties \cite{D'Hoker:2004ce,Cacciatori:2008ay}.
Such conditions impose strong constraints on the modular form
$\Xi^{(g)}[\delta](\Omega)$ defined in \eqref{twoloopmeas}, which,
at least for low genera, are sufficient to uniquely characterise
this form. It is easier to first describe the constraints satisfied
by $\Xi^{(g)}[0]$:
\begin{enumerate}
\item $\Xi^{(g)}[0](\Omega)$ is a modular form of weight $8$
under $\Gamma_g(1,2)\subset \Gamma_g=\Sp(2g,\ZZ)$
\be\label{ximod}
\Xi^{(g)}[0]((A\Omega+B)(C\Omega+D)^{-1})=\det(C\Omega+D)^8\Xi^{(g)}[0](\Omega)\
, \ee
$\left(\begin{smallmatrix} A & B\\ C & D
\end{smallmatrix}\right)\in\Gamma_g(1,2)$ (see appendix \ref{s:mathback} for more details on modular forms).

\item In the limit
$$ \Omega_g\to \begin{pmatrix}
\Omega_k & 0\\ 0 & \Omega_{g-k}
\end{pmatrix}\ , $$
where $\Omega_k\in\Sieg_k$, $\Omega_{g-k}\in\Sieg_{g-k}$,
$\Xi^{(g)}[0](\Omega_g)$ must factorise
\be\label{xifact} \Xi^{(g)}[0](\Omega_g)\to
\Xi^{(k)}[0](\Omega_k)\Xi^{(g-k)}[0](\Omega_{g-k})\ .\ee
\item For $g=1$, the known result for the chiral measure must be
reproduced, so that \be\label{xinorm}
\Xi^{(1)}[0](\tau)=\theta[0](\tau)^4\prod_{\delta\text{
even}}\theta[\delta](\tau)^4\ ,\ee with $\tau\in\Sieg_1$.
\end{enumerate}
Once these properties are satisfied for a certain $\Xi^{(g)}[0]$,
then, for any other even spin structure $\delta$ we can define
\be\label{xigendelta}
\Xi^{(g)}[\delta](\Omega):=\det(C\Omega+D)^{-8}\Xi^{(g)}[0]((A\Omega+B)(C\Omega+D)^{-1})\
, \ee where $\left(\begin{smallmatrix} A & B\\ C & D
\end{smallmatrix}\right)\in \Gamma_g$ satisfies (see eqs.\eqref{modulcharac} and \eqref{emmedelta})
\be\label{spintransf} \spin{\delta'}{\delta''}=\Bigl[0\cdot
\begin{pmatrix}
A & B\\ C & D
\end{pmatrix}\Bigr]=
\begin{bmatrix}
(\tp AC)_0\\(\tp BD)_0
\end{bmatrix}\ ,\ee (for any matrix $A$, we denote by $A_0$ the vector of diagonal entries). With this definition, each $\Xi^{(g)}[\delta]$ can be shown
to satisfy all the constraints from modular invariance and
factorisation, as an immediate consequence of \eqref{ximod},
\eqref{xifact} and \eqref{xinorm}.

Space-time supersymmetry implies that the cosmological constant must
vanish after the GSO projection. In terms of the chiral measure,
this condition becomes
\begin{enumerate}
\item[4.] Vanishing of the cosmological constant \be\label{xicosm} \sum_{\delta\text{ even}}\Xi^{(g)}[\delta](\Omega)=0\
,\qquad\qquad \Omega\in\J_g\subseteq \Sieg_g\ ,\ee where $\J_g$ is the
locus of the period matrices of Riemann surfaces of genus $g$.
\end{enumerate}
For $g\le 4$, this
last condition is a consequence of \eqref{ximod}, \eqref{xifact} and \eqref{xinorm}, while at genus $5$ it
must be imposed as an independent constraint. The solution of the above conditions in the case of hyperelliptic Riemann surfaces has
been found by Poor and Yuen \cite{PoorYuen}.

\subsection{Grushevsky ansatz}\label{s:grus}

In \cite{Grushevsky:2008zm} an ansatz has been proposed for the
chiral superstring measure which satisfies the conditions
\eqref{ximod}, \eqref{xifact} and \eqref{xinorm}. At genus $5$, a
modified version of this ansatz is needed to satisfy also
\eqref{xicosm} \cite{Grushevsky:2008zp}. In this section, we
describe Grushevsky's construction and prove that the functions
$\hat A_2$ and $\hat A_3$ defined in \eqref{twopvanish} and
\eqref{threepvanish} do not vanish at $g=4$.

Let $V$ be a vector subspace of $\FF_2^{2g}$, with $\FF_2:=\{0,1\}$
the field of characteristic $2$. Set
$$P(V):=\prod_{\delta'\in V}\theta[\delta']\ ,\qquad P^g_{i,s}:=\sum_{V,\ \dim V=i}P(V)^s\ .$$
For each $\delta\in\FF_2^{2g}$, consider the affine
space $A:=\delta+V$ and define
$$P(A)\equiv P(V+\delta):=\prod_{\delta'\in V}\theta[\delta'+\delta]\ ,$$
and
$$P^g_{i,s}[\delta]:=\sum_{V,\ \dim V=i}P(V+\delta)^s=
\sum_{A\ni\delta,\ \dim A=i,}P(A)^s\ .
$$
Grushevsky proposal for the modular form $\Xi^{(g)}[\delta]$
appearing in the superstring measure
$d\mu[\delta]=\Xi^{(g)}[\delta]d\mu_{Bos}$ is
$$\Xi^{(g)}_{G}[\delta]:=2^{-g}\sum_{i=0}^g(-1)^i2^{\frac{i(i-1)}{2}}P^g_{i,2^{4-i}}[\delta]\ .$$

\bigskip

\noindent The cosmological constant is
$$\Xi_G^{(g)}:=\sum_{\delta\text{ even}}\Xi_G^{(g)}[\delta]
=2^{-g}\sum_{i=0}^g(-1)^i2^{\frac{i(i+1)}{2}}\sum_{A,\ \dim
A=i}P(A)^{2^{4-i}}=2^{-g}\sum_{i=0}^g(-1)^i2^{\frac{i(i+1)}{2}}S_{i,2^{4-i}}\ , $$ where
$$S_{i,s}:=\sum_{A,\ \dim
A=i}P(A)^{s}\ .
$$
Note the factor $2^{\frac{i(i+1)}{2}}$ which differs from
$2^{\frac{i(i-1)}{2}}$ in the definition of $\Xi^{(g)}[\delta]$,
because in the cosmological constant each affine space $A$ of
dimension $i$ is counted $2^i$ times, one per each element
$\delta\in A$. The cosmological constant up to $g=5$ can be computed
using the following relations for the modular forms $S_{i,2^{4-i}}$
\cite{Igusa:1981}
\begin{align*}
(2^{2g}-1)S_{0,16}&=6S_{1,8}+24S_{2,4}\ , & g\ge 2\\
(2^{2g-2}-1)S_{1,8}&=18S_{2,4}+168S_{3,2}\ , & g\ge 3\\
(2^{2g-4}-1)S_{2,4}&=42S_{3,2}+840S_{4,1}\ , & g\ge 4
\end{align*}
together with the following relation which holds on
$\J_g\subseteq\Sieg_g$ for $g\ge 5$ \cite{Grushevsky:2008zp}
\begin{align*} (2^{2g-6}-1)S_{3,2}&=90S_{4,1}+3720S_{5,1/2}\ , &g\ge 5\ .
\end{align*} It follows that
$$\Xi^{(g)}_G=2^{g-1}(2^g+1)D_gJ^{(g)}\ ,
$$
for some non-vanishing $D_g\in\CC$, where $2^{g-1}(2^g+1)$ is the
number of even spin structures at genus $g$, \be\label{jayg}
J^{(g)}:=\Theta_{E_8}^2-\Theta_{D_{16}^+}=2^{-2g}((1-2^g)S_{0,16}+2S_{1,8})\
, \ee and $\Theta_{E_8}$, $\Theta_{D_{16}^+}$ are the theta series
corresponding to the even unimodular lattices $\Lambda=E_8$ and
$\Lambda=D_{16}^+$ (see subsection \ref{s:OPSMY}). In particular,
$\Xi^{(g)}_G=0$ for $g=2,3$, because $J^{(g)}$ vanishes identically
on $\Sieg_g$ for $g\le 3$, while \be\label{Dfour}
D_4=-\frac{2^{7}\cdot 3}{7\cdot 17}\ , \ee and 
$$D_5=-\frac{2^{11}\cdot 17}{7\cdot 11\cdot 31}\ . $$
For $g=4$ the form $J^{(4)}$ vanishes identically on the locus
$\J_4$ (in fact, $\J_4$ is the divisor of $J^{(4)}$ inside
$\Sieg_4$), while $J^{(5)}\neq 0$ on $\J_5$
\cite{Grushevsky:2008zp}. Thus, for the constraint \eqref{xicosm} to
be satisfied, one has to introduce a modified measure at $g=5$
\be\label{corrGru} \tilde\Xi_G^{(5)}[0]:=\Xi_G^{(5)}[0]-D_5J^{(5)}\
, \ee that continues to satisfy the factorisation properties and
assures that eq.\eqref{xicosm} is satisfied.

\bigskip

Let $C$ be a Riemann surface of genus $g$ and define
\be\label{faketwo} \hat A_2(a,b):=\sum_{\delta\text{
even}}\Xi^{(g)}[\delta]S_\delta(a,b)^2\ , \ee $a,b\in C$, where
$S_\delta(a,b)$ is the Szeg\"o kernel (see appendix
\ref{appendiceuno}). It has been proposed that the chiral two-point
function for NS states on $C$ corresponds to $\hat A_2(a,b)$ up to
spin independent factors. By space-time supersymmetry, the two-point
function is expected to vanish identically on any Riemann surface.
It has been proved \cite{Grushevsky:2008qp} that with Grushevsky
ansatz this condition on \eqref{faketwo} is satisfied for $g\le 3$.
In the following, we will prove that such a condition does not hold
for $g=4$. This is an immediate consequence of the following useful
lemma.

\begin{lemma}
\be\label{sumdelta} \frac{\partial\Xi^{(g)}_G}{\partial\Omega_{ij}}=\frac{2^4}{2\pi
i(1+\delta_{ij})}\sum_{\delta}\Xi^{(g)}_G[\delta]\,\partial_i\partial_j\log\theta[\delta]
\ .\ee
\end{lemma}
\begin{proof} By a direct computation
\begin{equation*}\begin{split} \frac{\partial S_{k,s}}{\partial\Omega_{ij}}&=\sum_{A,\
\dim A=k}\frac{\partial}{\partial\Omega_{ij}}\prod_{\delta\in
A}\theta[\delta]^{s}=s\sum_{A,\ \dim A=k}\sum_{\delta\in
A}P(A)^s\frac{\partial}{\partial\Omega_{ij}}\log\theta[\delta]\\
&
=s\sum_{\delta}\sum_{V,\ \dim
V=k}P(\delta+V)^s\frac{\partial}{\partial\Omega_{ij}}\log\theta[\delta]\ ,
\end{split}\end{equation*} so that
\begin{equation*}\begin{split}\frac{\partial\Xi^{(g)}_G}{\partial\Omega_{ij}}&=2^{-g}\sum_{k=1}^g(-1)^k2^{\frac{k(k+1)}{2}}
\frac{\partial S_{k,2^{4-k}}}{\partial\Omega_{ij}}\\
&=2^{4-g}\sum_{\delta}\sum_{k=1}^g(-1)^k2^{\frac{k(k-1)}{2}}\sum_{V,\
\dim
V=k}P(\delta+V)^{2^{4-k}}\frac{\partial}{\partial\Omega_{ij}}\log\theta[\delta]\\
&=2^4\sum_{\delta}\Xi^{(g)}_G[\delta]\frac{\partial}{\partial\Omega_{ij}}\log\theta[\delta]=
2^4\sum_{\delta}\frac{\Xi^{(g)}_G[\delta]}{2\pi
i(1+\delta_{ij})}\partial_i\partial_j\log\theta[\delta]\ ,
\end{split}
\end{equation*}
where, in the last line, we used the heat equation for the theta
function
\be\label{heat} \partial_i\partial_j\theta[\delta](z,\Omega)=2\pi
i(1+\delta_{ij})\frac{\partial}{\partial\Omega_{ij}}\theta[\delta](z,\Omega)\
. \ee
\end{proof}

\begin{proposition}\label{th:faketwo} In the case $\Xi^{(g)}[\delta]$ in \eqref{faketwo} is identified with
$\Xi^{(g)}_G[\delta]$, we have
$$\hat A_2(a,b)=\omega(a,b)\Xi_G^{(g)}+\frac{2\pi i}{16}
\sum_{i\le
j}^g\frac{\partial\Xi_G^{(g)}}{\partial\Omega_{ij}}(\omega_i(a)\omega_j(b)+\omega_i(b)
\omega_j(a))\ .
$$
\end{proposition}

\begin{proof}
First use the relation (formula 38 page 25 of \cite{Fay:1973}, see
also appendix \ref{a:twotheta} for a proof)
\be\label{questa}S_\delta(a,b)^2=\omega(a,b)+\sum_{i,j}^g\omega_i(a)\omega_j(b)\partial_i\partial_j
\log\theta[\delta](0)\ , \ee to obtain \be\label{A2formula} \hat
A_2[\delta](a,b)\equiv \Xi^{(g)}[\delta]S_\delta(a,b)^2=
\Xi^{(g)}[\delta]\omega(a,b)+\sum_{i,j}^g\Xi^{(g)}[\delta]\omega_i(a)\omega_j(b)\partial_i\partial_j
\log\theta[\delta](0)\ , \ee
 then use
the previous lemma.
\end{proof}

\noindent {}This result leads immediately to the known results for
$g\le 3$ and to a new one for $g=4$.
\begin{corollary}
For $g\le 3$, \be\label{faketwog3} \hat A_2(a,b)= 0\ ,\ee while for
$g=4$ \be\label{faketwog4} \hat A_2(a,b)d\mu^{(4)}_{Bos}=c\sum_{i\le
j}(-1)^{m_{ij}}
(\omega_i(a)\omega_j(b)+\omega_i(b)\omega_j(a))\bigwedge_{\substack{k\le
l,\,(k,l)\neq(i,j)}}^4d\Omega_{kl}\neq 0\ ,\ee for some non-zero
constant $c\in\CC$ and $m_{ij}\in\ZZ$.
\end{corollary}
\begin{proof} Eq.\eqref{faketwog3} follows immediately from proposition \ref{th:faketwo} and the fact that for $g\le 3$
$\Xi_G^{(g)}=0$ identically on $\Sieg_g$. As proved in
\cite{SalvatiManni:2008qa}, $\Xi_G^{(4)}=D_4J^{(4)}$, with $D_4\neq
0$ given in \eqref{Dfour}. Furthermore, $\frac{\partial
J^{(4)}}{\partial\Omega_{ij}}$ cannot vanish identically on $\J_4$
for all $i,j$, because $\J_4$ is the divisor of $J^{(4)}$ and is
irreducible \cite{Igusa:1982}. Fix some $1\le i,j\le 4$ and consider
the open subset of $\J_4$ where $\partial
J^{(4)}/\partial\Omega_{ij}\neq 0$. The bosonic string measure on
this subset is given (up to a constant) by
$$ d\mu^{(4)}_{Bos}=(-1)^{m_{ij}}\frac{
\bigwedge_{\substack{k\le
l,\,(k,l)\neq(i,j)}}^4d\Omega_{kl}}{\partial
J^{(4)}/\partial\Omega_{ij}}\ ,$$ where $m_{ij}$ is the position of
$d\Omega_{ij}$ with respect to a given ordering of
$\{d\Omega_{kl}\}_{k\le l}$. Then, for each point in $\J_4$, we have
\begin{equation*}\begin{split} \hat A_2(a,b)d\mu^{(4)}_{Bos}&=\frac{2\pi
iD_4}{16}d\mu^{(4)}_{Bos}\sum_{\substack{i\le j\\ \partial
J^{(4)}/\partial\Omega_{ij}\neq 0}} \frac{\partial
J^{(4)}}{\partial\Omega_{ij}}(\omega_i(a)\omega_j(b)+\omega_i(b)\omega_j(a))\\
&=\frac{2\pi iD_4}{16}\!\!\!\!\sum_{\substack{i\le j\\ \partial
J^{(4)}/\partial\Omega_{ij}\neq
0}}\!\!\!\!\! (-1)^{m_{ij}}
(\omega_i(a)\omega_j(b)+\omega_i(b)\omega_j(a))\bigwedge_{\substack{k\le l,\,(k,l)\neq(i,j)}}^4d\Omega_{kl}\\
&=\frac{2\pi
iD_4}{16}\sum_{i\le j}(-1)^{m_{ij}}
(\omega_i(a)\omega_j(b)+\omega_i(b)\omega_j(a))\bigwedge_{\substack{k\le
l,\,(k,l)\neq(i,j)}}^4d\Omega_{kl}\ ,
\end{split}\end{equation*}
where we used the fact that $d\Omega_{11}\wedge\ldots\wedge
\hat{d\Omega_{ij}}\wedge\ldots\wedge d\Omega_{44}=0$ when $\partial
J^{(4)}/\partial\Omega_{ij}=0$.
\end{proof}

This corollary also implies that the proposed three-point function
\be\label{fakethree} \hat A_3(a,b,c):=\sum_{\delta\text{
even}}\Xi^{(g)}[\delta]S_\delta(a,b)S_\delta(b,c)S_\delta(c,a)\ ,\ee
does not vanish for $g=4$, as expected from space-time
supersymmetry. To see this, note that, in the limit $c\to a$, the
coefficient of the term $(c-a)^{-2}$ coincides with $\hat A_2(a,b)$.
As discussed in \cite{Matone:2008td}, the fact that \eqref{faketwo}
and \eqref{fakethree} do not vanish for $g=4$ does not really rule
out the proposals $\Xi^{(4)}[\delta]$ for the chiral measure,
because it is reasonable that the two- and three-point functions
receive other contributions different from $\hat A_2$ and $\hat
A_3$. For the same reasons, however, the fact that $\hat A_2$
vanishes at $g=3$ cannot be considered as a real argument in favor
of this ansatz. In the following sections, we will consider a more
reliable computation for the two-point function based on the
factorisation of vacuum amplitudes.

\subsection{The OPSMY ansatz}\label{s:OPSMY}

In this section, we define $\Xi^{(g)}[\delta]$ in terms of theta
series of $16$-dimensional unimodular lattices, following Oura,
Poor, Salvati Manni and Yuen (OPSMY) \cite{oura-2008}. A
$d$-dimensional lattice $\Lambda\subset\RR^d$ is called unimodular
(or self-dual) if it is isomorphic to its dual
$\Lambda\cong\Lambda^*$, where
$$ \Lambda^*:=\{\lambda\in\RR^d\mid \lambda\cdot \mu\in\ZZ\text{ for all
}\mu\in\Lambda\}\ .$$ A unimodular lattice is called even if the
norm $\lambda\cdot\lambda$ of all its vectors is an even integer,
and odd otherwise. There are eight $16$-dimensional unimodular
lattices, listed in table \ref{t:lattices}, where $E_8^2$ and
$D_{16}^+$ are even and the others odd \cite{ConwaySloane}. The genus $g$ theta series
of a lattice $\Lambda$ is a holomorphic function on $\Sieg_g$ defined as
\be\label{thetaseries}
\Theta^{(g)}_\Lambda (\Omega):=\sum_{\lambda_1,\ldots,\lambda_g\in\Lambda}e^{\pi
i \sum_{i,j}^g\lambda_i\cdot\lambda_j \Omega_{ij}} \ .\ee

\begin{table}
\begin{center}
\begin{tabular}{|c|c|c|c|c|c|c|c|c|}\hline $k$ &
$\Lambda_k$  & parity & $n_k$ & $\g_k=\tilde\g_k\oplus D_{n_k}$ &
$l_k$ & $N_k$ & $\Lambda_k^{(1)}$ & $\Lambda_k^{(2)}$ \\\hline
 $0$ & $(D_8\oplus D_8)^+$ & odd & $0$ & $(D_8\oplus D_8)\oplus 0$ & $28$ & 224 & $D_{16}^+$ & $E_8^2$ \\\hline
 $1$ & $\ZZ\oplus A_{15}^+$ & odd & 1 & $A_{15}\oplus 0$ & 32 & 240 & $D_{16}^+$ & $D_{16}^+$ \\\hline
 $2$ & $\ZZ^2\oplus(E_7\oplus E_7)^+$ & odd & 2 & $2E_7\oplus 2A_1 $ & 36 & 256 & $E_8^2$ & $E_8^2$ \\\hline
 $3$ & $\ZZ^4\oplus D_{12}^+$ & odd & 4 & $D_{12}\oplus D_4$ & 44 & 288 & $D_{16}^+$ & $D_{16}^+$ \\\hline
 $4$ & $\ZZ^8\oplus E_8$ & odd & 8 & $E_8\oplus D_8$ & 60 & 352 & $E_8^2$ & $E_8^2$ \\\hline
 $5$ & $\ZZ^{16}$ & odd & 16 & $0\oplus D_{16}$ & 92 & 480 & $D_{16}^+$ & $D_{16}^+$ \\\hline
 $6$ & $E_8\oplus E_8$ & even & 0 & $(E_8\oplus E_8)\oplus 0$ & 60 & 480 & $E_8^2$ & $E_8^2$ \\\hline
 $7$ & $D_{16}^+$ & even & 0 & $D_{16}\oplus 0$ & 60 & 480 & $D_{16}^+$ & $D_{16}^+$ \\\hline
\end{tabular}
\end{center}\caption{\small The $16$-dimensional unimodular lattices.
The vectors of norm $2$ form the root system of the Lie algebra
$\g_k$. Each lattice can be decomposed as
$\Lambda_k=\tilde\Lambda_k\oplus\ZZ^{n_k}$, where $\tilde\Lambda_k$
has minimal norm $2$, and the associated Lie algebras decompose
accordingly $\g_k=\tilde\g_k\oplus D_{n_k}$ (with the identification
$D_2\equiv A_1\oplus A_1$). $l_k$ is twice the dual Coxeter number of
$\tilde\g_k$ and $N_k$ is the number of roots in $\g_k$. The value
$l_5=92$ is chosen for later convenience. See appendix
\ref{a:thetaseries} for the definition of
$\Lambda_k^{(1)},\Lambda_k^{(2)}$.} \label{t:lattices}\end{table}

\noindent Following \cite{oura-2008}, let us define $\xi^j:=
(\xi^j_0,\ldots,\xi^j_5)\in\CC^6$, $j=0,\ldots,5$, by\footnote{We
use a different normalisation with respect to \cite{oura-2008}, so
that $\Xi^{(g)}_{OPSMY}[\delta]$ and $\Xi^{(g)}_G[\delta]$ have the
same normalisation. In particular, $c^i_k$ is $2^{4i}$ times the corresponding coefficient in \cite{oura-2008}.}
$$\textstyle{ \xi^0:=(1,1,1,1,1)\ , \qquad \qquad \xi^j:=(0,\frac{1}{8^j},\frac{1}{4^j},\frac{1}{2^j},1,2^j)\ ,\quad j=1,\ldots,5\
, }$$ and the dual basis
$\{(c^i_0,\ldots,c^i_5)\}_{i=0,\ldots,5}\subset\CC^6$, so that
$$ \sum_{k=0}^5 c^i_k\xi^j_k=\delta_{ij}\ .
$$ For $g< 5$, the theta series of the $16$-dimensional unimodular lattices
$$
\Theta_k^{(g)}:=\Theta_{\Lambda_k}^{(g)}\ ,$$ $k=0,\ldots,7$, are
not linearly independent and the linear relations can be easily
expressed using the coefficients $c_k^i$. In particular
\cite{oura-2008}, \be\label{reltheta3} \sum_{k=0}^5
c^i_k\Theta^{(g)}_k=0\ , \qquad \text{for }g\le 3,\quad g<i\le 5\ ,
\ee \be\label{reltheta4} \sum_{i=0}^5 c^5_k\Theta^{(4)}_k=CJ^{(4)}\
,\ee where $J^{(g)}$ is defined in \eqref{jayg} and \be\label{Ccost}
C=-\frac{2^5\cdot 3}{7}\ . \ee There are also well-known relations
between the theta series of even lattices \be\label{relsch}
J^{(g)}=0\ ,\quad g\le 3\ ,\qquad\qquad J^{(4)}_{\rvert \J_4}=0\
.\ee

\bigskip

\noindent Set
$$
\Xi^{(g)}_{OPSMY}[0](\Omega):=\sum_{k=0}^5c^g_k\Theta_{k}^{(g)}(\Omega)\
. $$ By \eqref{xigendelta}, $\Xi_{OPSMY}^{(g)}[\delta]$ for every
even $\delta\in\FF_2^{2g}$ can be easily expressed in terms of the
corresponding lattice theta series
\be\label{defseries}\Theta_k^{(g)}[\delta](\Omega)=\det(C\Omega+D)^{-8}\,\Theta_k^{(g)}\bigl((A\Omega+B)(C\Omega+D)^{-1}\bigr)
\ , \ee where $\left(\begin{smallmatrix} A & B\\ C & D
\end{smallmatrix}\right)$ satisfies \eqref{spintransf}. Therefore
$$
\Xi^{(g)}_{OPSMY}[\delta](\Omega)=\sum_{k=0}^5c^g_k\Theta_{k}^{(g)}[\delta](\Omega)\ . $$
A useful expression for $\Theta_\Lambda^{(g)}[\delta](\Omega)$ is (see
appendix \ref{a:thetaseries} for a derivation)
\be\label{thserexpr}\begin{split}
\Theta_\Lambda^{(g)}[\delta](\Omega)=\sum_{\lambda_1,\ldots,\lambda_g\in
\Lambda}\exp\Bigl(\pi i
\sum_{i,j}^g(\lambda_i+\frac{\delta'_iu}{2})\cdot(\lambda_j+\frac{\delta'_ju}{2})
\Omega_{ij}+2\pi i \sum_{i}^g(\lambda_i+\frac{\delta'_iu}{2})\cdot
(\delta''_i\frac{u}{2})\Bigr)\ ,\end{split}\ee where $u\in\Lambda$
is a parity vector for $\Lambda$, i.e. $u\cdot\lambda\equiv
\lambda\cdot\lambda\mod 2$ for all $\lambda\in\Lambda$. We take
\eqref{thserexpr} to be the definition of $\Theta_{\Lambda}^{(g)}[\delta]$
for a general integral lattice $\Lambda$ and for every theta characteristic $\delta\in\FF_2^{2g}$. Note that, even though
this definition makes sense more generally, eq.\eqref{defseries}
holds for some $\left(\begin{smallmatrix}A & B\\ C & D\end{smallmatrix}\right)\in \Sp(2g,\ZZ)$ only if $\Lambda$ is a $16$-dimensional unimodular lattice and $\delta$ is even.

\bigskip

Summing $\Xi_{OPSMY}^{(g)}[\delta]$ over spin structures yields (see
appendix \ref{a:thetaseries} for the derivation)
$$ \sum_{\delta\text{
even}}\Xi^{(g)}_{OPSMY}[\delta](\Omega)=2^{g-1}(2^g+1)B_gJ^{(g)}\
,$$ where
\be\label{Bfour} B_4=\frac{2^2\cdot 3^3\cdot 5\cdot 11}{7\cdot 17}\ ,
\ee and
$$B_5=-\frac{2^5\cdot 17}{7\cdot 11}\ .$$
Thus, in analogy with \eqref{corrGru}, we define a modified
measure for $g=5$
$$ \tilde\Xi^{(5)}_{OPSMY}[0](\Omega):=\sum_{k=0}^7c^5_k\Theta_{k}^{(5)}(\Omega)=\Xi_{OPSMY}^{(5)}[0](\Omega)
-B_5J^{(5)}\ , $$ where, for all $g$ (this modification is
irrelevant for $g\le 4$), we set
\be\label{csixcseven} c^g_6=-c^g_7= 
-B_5\ ,\ee so that eq.\eqref{xicosm} is satisfied.

\medskip

It is known that two forms $\Xi^{(g)}$, satisfying
\eqref{ximod},\eqref{xifact} and \eqref{xinorm}, must be the same for
$g=2,3$ while for $g=4$ differ by a multiple of the Igusa-Schottky form
$J^{(4)}$, which vanishes on the Jacobian locus, so that, by
\eqref{Dfour} and \eqref{Bfour},
\be\label{dueventisei} \Xi^{(g)}_{OPSMY}[\delta]=\Xi^{(g)}_{G}[\delta]\ ,\quad g\le 3\
,\qquad\qquad\qquad
\Xi^{(4)}_{OPSMY}[\delta]=\Xi^{(4)}_{G}[\delta]+(B_4-D_4)J^{(4)}\ .\ee

It is an open question whether $\tilde \Xi^{(5)}_{OPSMY}[\delta]$ and
$\tilde\Xi^{(5)}_{G}[\delta]$ coincide on the Jacobian locus (see
\cite{DuninBarkowski:2009ej} for a discussion on this point). From
now on, we will drop the subscripts in $\Xi^{(g)}_G$ and
$\Xi^{(g)}_{OPSMY}$ when the result is independent of the particular
definition.

\section{Two-point function from factorisation}

Consider a family of Riemann surfaces $\tilde C_t$, $0<|t|<1$ of genus $g+1$ such that,
in the limit $t\to 0$, one of the handles becomes infinitely long or, in the conformally equivalent picture,
the cycle $\tilde\alpha_{g+1}$ around this handle is pinched to form a node. This family of surfaces can be defined using the standard plumbing fixture procedure, see subsection \ref{a:degen}.

It is an old idea, both in conformal field theory and string theory,
that, in the limit $t\to 0$, the amplitudes defined on $\tilde C_t$
must satisfy suitable factorisation properties
\cite{Friedan:1986ua}. Let us give a rough description of the
physical picture behind this idea in the simple case of the
zero-point amplitude $Z$ in some conformal field theory. A genus
$g+1$ Riemann surface $\tilde C_t$ can be obtained by ``gluing'' a
long thin cylinder to a Riemann surface $ C$ of genus $g$ with two
holes. In the limit $t\to 0$, the cylinder becomes infinitely thin
and the holes collapse to two punctures $a,b\in C$. Then, the
boundary conditions of the fields around these punctures can be
described in terms of vertex operators at $a$ and $b$. In a
state-operator formalism, the propagation of states along the
cylinder is given by an operator $t^{L_0}\bar t^{\bar L_0}$, where
$L_0$ and $\bar L_0$ generate the world-sheet dilatations and
rotations. Thus, the zero-point function $Z$ can be expanded as
\be\label{partitionf} Z \to  \sum_{\phi} \langle \phi| t^{L_0}\bar
t^{\bar L_0}|\phi\rangle\langle V(\phi,a) V(\phi^*,b)\rangle_C\ ,
\ee where the sum runs over a complete set of states of the theory
and $\phi^*$ denote the conjugated of $\phi$.

\bigskip

In the following, we will apply this procedure to obtain the (spin dependent part of the) chiral two-point function for two
NS massless states on a surface of genus $g$ from factorisation of the chiral measure at genus $g+1$. Specialising
eq.\eqref{partitionf} to the case where $Z$ is the superstring chiral zero-point function before GSO projection, one obtains
$$ d\mu^{(g+1)}[\tilde\delta]\to \sum_{\phi} t^{h_\phi}\langle V(\phi,a) V(\phi^*,b)\rangle_C \ ,
$$ where the sum is over a complete set of chiral $L_0$-eigenstates with $L_0\phi=h_\phi\phi$. The spin structure $\tilde\delta\in\FF_2^{2g+2}$
is given by $$\tilde\delta=\spin{\tilde\delta'_1 & \ldots & \tilde\delta'_{g+1}}{\tilde\delta''_1 & \ldots & \tilde\delta''_{g+1}}=
\spin{\delta'_1 & \ldots & \delta'_g & \epsilon'}{\delta''_1 & \ldots & \delta''_g & \epsilon''} \in\FF_2^{2g+2}\ ,$$ where each $\tilde\delta_i'$
(respectively, $\tilde\delta_i''$), $i=1,\ldots,g+1$, determines the periodicity of the world-sheet fermionic fields around the cycle $\tilde\alpha_i$ (resp., $\tilde\beta_i$) of $\tilde C_t$.
In particular, the cycle $\tilde\alpha_{g+1}$ encircles the infinitely long cylinder in the limit $t\to 0$, so that, when $\epsilon'=0$
(respectively, $\epsilon'=1$), the sum in \eqref{partitionf} runs only over the Neveu-Schwarz (resp., Ramond) sector.
Thus, the two-point function at genus $g$ for NS states and for an arbitrary even spin structure $\delta\in\FF_2^{2g}$  can be obtained from the degeneration limit of
$$ d\mu^{(g+1)}\thspin{\delta' & 0}{\delta'' & 0}\qquad\qquad\text{or}\qquad\qquad d\mu^{(g+1)}\thspin{\delta' & 0}{\delta'' & 1}\ .
$$  To project out the NS tachyon, we consider a linear combination such that its leading term corresponds to massless states.
Because the GSO projection is implemented by summing the chiral
measure over all spin structures without phases, the correct linear
combination to consider is
$$ \frac{1}{2}\Bigl(d\mu^{(g+1)}\thspin{\delta' & 0}{\delta'' & 0}+d\mu^{(g+1)}\thspin{\delta' & 0}{\delta'' & 1}\Bigr)=d\mu^{(g+1)}_{Bos}
X_{NS}\thspin{\delta'}{\delta''}\ ,
$$
where
$$ X_{NS}\thspin{\delta'}{\delta''}:=
\frac{1}{2}\Bigl(\tilde\Xi^{(g+1)}\thspin{\delta'&0}{\delta''&0}
+\tilde\Xi^{(g+1)}\thspin{\delta'&0}{\delta''&1}\Bigr)\
,$$ is the spin dependent part. Indeed, the tachyon contribution would correspond to the
$t^{1/2}$ power in the expansion of $X_{NS}$, but it can be verified that all half-integer powers of $t$ are canceled
by summing over the two spin structures. It follows that,
as $t\to 0$,
\be\label{XNSexp} X_{NS}[\delta]=t \frac{A_2[\delta](a,b)}{d z(a)\, dz(b)}+O(t^2)\ ,
\ee up to an irrelevant spin-independent factor, where $A_2[\delta]$ is the chiral two-point function for NS
massless states. Note that $A_2[\delta](a,b)$ is a meromorphic $1$-differential in $a,b$ and $A_2[\delta](a,b)/(d z(a)dz(b))$ corresponds to its evaluation
in the local coordinates around $a$ and $b$ used in the plumbing fixture construction, see subsection \eqref{a:degen}.

\subsection{Computation of the two-point function}

Here we compute $A_2[\delta](a,b)$ with the mentioned coordinate choice (see eq.\eqref{XNSexp}), using the OPSMY ansatz for the chiral measure at genus $g+1$, for $g=2,3,4$.
As we will see, such a choice of the local coordinate at the node of degenerate Riemann
surfaces leads to a considerable simplification of the calculations.

For a general $\tilde\Omega\in \Sieg_{g+1}$ define
$$ \tilde\Omega=\begin{pmatrix} \Omega & \tp z\\ z &
\frac{1}{2\pi i}\log q \end{pmatrix} \ , $$ where $\Omega\in\Sieg_g$, $z\in\CC^g$ and
$q\in\CC$, $0<|q|<1$. Consider a Riemann surface of genus $g+1$ and take  the degeneration
limit in which the cycle $\tilde\alpha_{g+1}$ is pinched. One obtains a
singular surface of genus $g$ with two points $a,b$ identified to
form a node. Let $t$ be the degeneration parameter and
$\tilde\Omega(t)\in \Sieg_{g+1}$ the corresponding period matrix.
 As $t\to 0$, for a suitable choice of local coordinates (see appendix \ref{a:degen}), we have
\begin{align} q(t):=&
=t+O(t^2)\ ,\qquad z_i(t)=\int_a^b\omega_i+O(t^2)\
,\quad i=1,\ldots,g\ ,\\
\Omega_{ij}(t)=&\Omega_{ij}+2\pi it
E(a,b)^2(\omega_i(a)\omega_j(b)+\omega_i(b)\omega_j(a))+O(t^2)\
,\quad i,j=1,\ldots,g\ ,\end{align} where $\Omega\in\J_g$. With respect to this choice of local coordinates, \eqref{XNSexp} becomes
\be\label{XNSexp2} X_{NS}[\delta]= t E(a,b)^2 A_2[\delta](a,b)+O(t^2)\ .
\ee
For a generic $\tilde\Omega\in\Sieg_{g+1}$, let us take the
expansion of $X_{NS}[\delta]\equiv X_{NS}[\delta](q,z,\Omega)$ as $q\to 0$
$$ X_{NS}[\delta](q,z,\Omega)= G^{(g)}[\delta](\Omega)
+q F^{(g)}[\delta](\Omega,z)+O(q^2)\ .$$ The effect of summing over
the two spin structures $\thspin{\delta'&0}{\delta''&0}$ and
$\thspin{\delta'&0}{\delta''&1}$ is to project out the half-integer
powers in $q$. Modular properties of
$\tilde\Xi^{(g+1)}\thspin{\delta'&0}{\delta''&0}(\tilde\Omega)$ and
$\tilde\Xi^{(g+1)}\thspin{\delta'&0}{\delta''&1}(\tilde\Omega)$ imply that
$G^{(g)}[\delta](\Omega)$ is independent of $z$ and
$F^{(g)}[\delta](\Omega,z)$ is a section of $|2\Theta|$ (see
appendix \ref{a:twotheta}). It follows that
$G^{(g)}[\delta](\Omega)$ can be computed for $z=0$. The
factorisation properties of the theta series
$$\Theta_k^{(g+1)}\thspin{\delta'&0}{\delta''&*}(\tilde\Omega)\stackrel{z\to 0}{\longrightarrow}
\Theta_k^{(1)}\thspin{0}{ *
}((2\pi i)^{-1}\log q)\Theta_k^{(g)}\thspin{\delta'}{\delta''}(\Omega)\
,$$ yield
$$
X_{NS}[\delta](q,z=0,\Omega)=\sum_{k=0}^7c^{g+1}_k(1+N_kq+O(q^2))\Theta_k^{(g)}[\delta](\Omega)\
,$$ where $N_k$ is the number of vectors of norm $2$ in the lattice
$\Lambda_k$ (see table \ref{t:lattices}). By \eqref{reltheta3},
\eqref{reltheta4}, \eqref{relsch} and \eqref{csixcseven}, it follows
that
$$ G^{(g)}[\delta](\Omega)=\sum_{k=0}^7c_k^{g+1}\Theta_k^{(g)}[\delta](\Omega)=(C-B_5)J^{(g)}(\Omega)\
. $$ Similarly, one can find the expansion of
$(c^{g+1}_0N_0,\ldots,c^{g+1}_5N_5)\in\CC^6$ with respect to the
basis $c^0,\ldots,c^5$ by computing the scalar products with the
dual basis $\xi^0,\ldots,\xi^5$
$$ \sum_{k=0}^5 c^{g+1}_kN_k\xi_k^{i}=0\ ,\ i<g\ ,\qquad
\sum_{k=0}^5 c^{g+1}_kN_k\xi_k^{g}=128\ ,\qquad \sum_{k=0}^5
c^{5}_kN_k\xi_k^{5}=720\ ,$$ to obtain (note that $N_6=N_7=480$)
\be\label{Fgzero}
F^{(g)}[\delta](\Omega,0)=128\Xi^{(g)}_{OPSMY}[\delta](\Omega)+(720C-480B_5)J^{(g)}\
.\ee It follows that, as $t\to 0$,
$$
X_{NS}[\delta]=t\Bigl(\sum_{i,j}^g2\pi i
E(a,b)^2\omega_i(a)\omega_j(b)(1+\delta_{ij})(C-B_5)\frac{\partial
J^{(g)}}{\partial\Omega_{ij}}+F^{(g)}[\delta](\Omega,b-a)\Bigr)+O(t^2)\ .$$ Since $F^{(g)}$ is a section of $|2\Theta|$, we can use the
standard result (see appendix \ref{a:twotheta} for a proof)
$$
F^{(g)}[\delta](\Omega,b-a)=E(a,b)^2\Bigl(F^{(g)}[\delta](\Omega,0)\omega(a,b)
+\frac{1}{2}\sum_{i,j}^g\partial_i\partial_jF^{(g)}[\delta](\Omega,0)\omega_i(a)\omega_j(b)\Bigr)\
. $$ When $\Omega\in\J_g$  eq.\eqref{Fgzero} is equivalent to
$F^{(g)}[\delta](\Omega,0)=128\Xi^{(g)}[\delta](\Omega)$ that, by \eqref{dueventisei}, holds in both cases
$\Xi^{(g)}[\delta](\Omega)\equiv\Xi^{(g)}_{OPSMY}[\delta](\Omega)$ and $\Xi^{(g)}[\delta](\Omega)\equiv\Xi^{(g)}_{G}[\delta](\Omega)$.
It follows that
\begin{equation*}\begin{split}
A_2[\delta](a,b)=&128\,\Xi^{(g)}[\delta](\Omega)\,\omega(a,b)\\&+
\sum_{i,j}^g\omega_i(a)\omega_j(b)\Bigl(2\pi
i(C-B_5)(1+\delta_{ij})\frac{\partial
J^{(g)}}{\partial\Omega_{ij}}+\frac{1}{2}\partial_i\partial_jF^{(g)}[\delta](\Omega,0)\Bigr)\
.\end{split} \end{equation*} Notice that, by \eqref{relsch}, the
term $\frac{\partial J^{(g)}}{\partial\Omega_{ij}}$ vanishes for
$g\le 3$ but not for $g=4$.

\bigskip

It remains to compute
$\partial_i\partial_jF^{(g)}[\delta](\Omega,0)$. For a general
lattice $\Lambda$, set
\begin{equation*}\begin{split} F^{(g)}_\Lambda[\delta](\Omega,z):=
&\sum_{\substack{\lambda_1,\ldots,\lambda_g\\
\lambda_k\in\Lambda+u\frac{\delta'_k}{2}}}
\sum_{\tilde\lambda\cdot\tilde\lambda=2} e^{\pi i
\lambda_k\cdot\lambda_l\Omega_{kl} +2\pi i
\sum_k\lambda_k\cdot(\tilde\lambda z_k+ u \frac{\delta''_k}{2})}\
\end{split}\ ,\end{equation*} and, in particular,
$F_k[\delta](\Omega,z)\equiv F_{\Lambda_k}[\delta](\Omega,z)$,
$k=0,\ldots,7$, with $\Lambda_k$ listed in table \ref{t:lattices}.
Note that
$F_\Lambda[\delta](\Omega,0)=N_\Lambda\Theta_\Lambda[\delta](\Omega)$,
where $N_\Lambda$ is the number of vectors of norm $2$ in $\Lambda$.
Since, by \eqref{thserexpr},
$$
\frac{1}{2}(\Theta^{(g+1)}_k\thspin{\delta'&0}{\delta''&0}(\tilde\Omega)+\Theta^{(g+1)}_k\thspin{\delta'&0}{\delta''&1}(\tilde\Omega))=
\Theta^{(g)}_k[\delta](\Omega)+qF^{(g)}_k[\delta](\Omega,z)+O(q^2)\
,$$ we have
$$
F^{(g)}[\delta](\Omega,z)=\sum_{k=0}^7c^{g+1}_kF^{(g)}_k[\delta](\Omega,z)\
, $$ and
\begin{equation*}\begin{split}
\partial_i\partial_jF_\Lambda[\delta](\Omega,0)=
&\sum_{\substack{\lambda_1,\ldots,\lambda_g\\
\lambda_k\in\Lambda+u\frac{\delta'_k}{2}}}(2\pi i)^2
\sum_{\tilde\lambda\cdot\tilde\lambda=2}(\lambda_i\cdot\tilde\lambda)
(\tilde\lambda\cdot\lambda_j) e^{\pi i
\sum_{k,l}\lambda_k\cdot\lambda_l\Omega_{kl} +\pi i \sum_k\lambda_k\cdot u
\delta''_k}\end{split}\ . \end{equation*} In general, the lattice $\Lambda_k$ is
a direct sum $\Lambda_k=\tilde\Lambda_k\oplus\ZZ^{n_k}$, where
$\tilde\Lambda_k$ has no vectors of norm $1$ \cite{ConwaySloane}. It follows that the
set of vectors of norm $2$ in $\Lambda$ splits into a disjoint union
$$ \{\lambda\in\Lambda\mid \lambda\cdot\lambda=2\}=\{\lambda\in\tilde\Lambda\mid
\lambda\cdot\lambda=2\}\sqcup \{\lambda\in \ZZ^n\mid\lambda\cdot\lambda=2\}\
.$$ Hence,
$$
F_{\Lambda}[\delta](\Omega,z)=\Theta_{\tilde\Lambda}[\delta](\Omega)F_{\ZZ^n}[\delta](\Omega,z)
+\Theta_{\ZZ^n}[\delta](\Omega)F_{\tilde\Lambda}[\delta](\Omega,z)\
.$$ The vectors of norm $2$ in $\tilde\Lambda_k$ are the roots of
a semi-simple Lie algebra $\tilde\g_k$ (see table \ref{t:lattices}).
Let $\Delta$ be the set of roots of a simple Lie algebra of rank
$r$, a standard result is
$$\sum_{\alpha\in\Delta}\alpha\tp\alpha=l_\Delta\II_r\ ,
$$ where $l_\Delta$ is a constant depending on the Lie algebra.
This can be proved by noting that the matrix on the left hand side
is invariant under the action of the Weyl group, so that it must be
proportional to the identity. The constant $l_\Delta$ can be easily
computed by taking the trace of both sides
$$ \sum_{\alpha\in\Delta}\alpha\cdot \alpha=rl_\Delta\ .
$$ In the case of simply-laced algebras one obtains
$$ l_\Delta=\frac{2N}{r}\ ,
$$ where $N$ is the number of roots (more generally, $l_\Delta$ is
twice the dual Coxeter number of the algebra). The Lie algebra
$\tilde g_k$ associated to $\tilde\Lambda_k$ is either simple or the
sum of two copies of the same simple algebra, so that
$$ \sum_{\tilde\lambda\cdot\tilde\lambda=2}(\lambda_i\cdot\tilde\lambda)(\tilde\lambda\cdot\lambda_j)
=l_k\, \lambda_i\cdot\lambda_j\ ,$$ with $l_k$ given in table
\ref{t:lattices}. From this identity, one sees that $\tilde F_k:=
F_{\tilde\Lambda_k}$ satisfies an analog of the heat-kernel equation
\eqref{heat}
$$\partial_i\partial_j\tilde F_k[\delta](\Omega,0)=2\pi
i(1+\delta_{ij})l_k\frac{\partial}{\partial\Omega_{ij}}\tilde\Theta_k[\delta](\Omega)\
. $$ Furthermore, since $\theta[\delta](\Omega,z)$ is even in $z$,
one has
\begin{equation*}\begin{split}
F_{\ZZ^n}[\delta](\Omega,z)&=\sum_{\substack{\tilde\lambda\in\ZZ^n\\
\tilde\lambda\cdot\tilde\lambda=2}}\prod_{i=1}^n
\theta[\delta](\Omega,\tilde\lambda_i z)=
2n(n-1)\theta[\delta](\Omega,0)^{n-2}\theta[\delta](\Omega,z)^2\ ,
\end{split}\end{equation*}
and
$$ \partial_i\partial_jF_{\ZZ^n}[\delta](\Omega,0)=
4n(n-1)\theta[\delta](\Omega,0)^{n-1}\partial_i\partial_j\theta[\delta](\Omega,0)=
2\pi
i(1+\delta_{ij})(4n-4)\frac{\partial}{\partial\Omega_{ij}}\Theta_{\ZZ^n}[\delta](\Omega)
\ . $$ Using these results, one gets
$$
\partial_i\partial_jF_k^{(g)}[\delta](\Omega,0)=2\pi i(1+\delta_{ij})
l_k\frac{\partial\Theta_k^{(g)}[\delta](\Omega)}{\partial\Omega_{ij}}+(4n_k-4-l_k)n_k
\Theta_k^{(g)}[\delta](\Omega)\partial_i\partial_j\log\theta[\delta](\Omega,0)\
,
$$ so that
\begin{equation*}\begin{split}
\partial_i\partial_jF^{(g)}[\delta](\Omega,0)=&\sum_{k=0}^7c^{g+1}_k\partial_i\partial_jF_{k}^{(g)}[\delta](\Omega,0)\\=&2\pi
i(1+\delta_{ij})\frac{\partial}{\partial\Omega_{ij}}\Bigl(\sum_{k=0}^7s_k^g
\Theta_{k}^{(g)}[\delta](\Omega)\Bigr)-\Bigl(\sum_{k=0}^5t_k^{g}\Theta_{k}^{(g)}[\delta](\Omega)\Bigr)
\partial_i\partial_j\log\theta[\delta](\Omega,0)\end{split}\end{equation*} where
$$ s_k^g:=c_k^{g+1}l_k\ , \qquad\qquad t_k^g:=c_k^{g+1}n_k(l_k-4n_k+4)\ .
$$
By an explicit computation, one can verify that
$$ \sum_{k=0}^5\xi_k^is^g_k=0\ ,\quad i<g\ ,\qquad\qquad
\sum_{k=0}^5\xi_k^gs^g_k=32 \ , \qquad\qquad
\sum_{k=0}^5\xi_k^5s^4_k=152\ ,$$ for $g=2,3,4$. By \eqref{reltheta3} and
\eqref{reltheta4}, and noting that $l_6=l_7=60$, one obtains
$$ \sum_{k=0}^7s_k^g
\Theta_{k}^{(g)}[\delta](\Omega)=32\Xi^{(g)}_{OPSMY}[\delta](\Omega)+(152C-60B_5)J^{(g)}\
. $$ Analogously,
$$\sum_{k=0}^5\xi_k^it^g_k=0\ ,\quad i<g\ ,\qquad\qquad
\sum_{k=0}^5\xi_k^gt^g_k=256\ ,$$ so that
$$
\sum_{k=0}^5t_k^{g}\Theta_{k}^{(g)}[\delta](\Omega)=256\Xi^{(g)}[\delta](\Omega)\
,\qquad\qquad \Omega\in\J_g\ .$$ The final expression for the chiral
two-point function is \begin{subequations}
\begin{align}\label{twofinal} A_2[\delta]&(a,b)=128\,\Xi^{(g)}[\delta](\Omega)\,\omega(a,b)\notag 
+\sum_{i,j}^g\omega_i(a)\omega_j(b)
\Bigl[-\frac{256}{2}\Xi^{(g)}[\delta](\Omega)
\partial_i\partial_j\log\theta[\delta](\Omega,0)\\&+2\pi
i(1+\delta_{ij})\frac{\partial}{\partial\Omega_{ij}}\Bigl(\frac{32}{2}\Xi^{(g)}_{OPSMY}[\delta](\Omega)
+\frac{1}{2}(152C-60B_5)J^{(g)}+(C-B_5)J^{(g)}\Bigr) \Bigr]\\
\label{twofinal2}&\phantom{(a,b)} =128 \hat A_2[\delta](a,b)+\sum_{i,j}^g\omega_i(a)\omega_j(b)
\Bigl[-256\Xi^{(g)}[\delta](\Omega)
\partial_i\partial_j\log\theta[\delta](\Omega,0)\notag\\&\phantom{(a,b)}\quad+2\pi
i(1+\delta_{ij})\frac{\partial}{\partial\Omega_{ij}}\Bigl(16\Xi^{(g)}_{OPSMY}[\delta](\Omega)
+(77C-31B_5)J^{(g)}\Bigr) \Bigr] \ ,\end{align}\end{subequations}
where eq.\eqref{A2formula} has been used. We stress that, in the
last line of \eqref{twofinal} and \eqref{twofinal2},
$\Xi^{(g)}_{OPSMY}[\delta]$ has been used instead of
$\Xi^{(g)}_{G}[\delta]$. This is important for $g=4$, because the
difference is proportional to $J^{(4)}$, whose derivatives $\partial
J^{(4)}/\partial\Omega_{ij}$ in the directions transverse to the
Jacobian locus $\J_4$ are not zero. Such an issue does not arise for
$\Xi^{(g)}[\delta]$ on the first line of \eqref{twofinal} and
\eqref{twofinal2}, since
$\Xi^{(4)}_{OPSMY}[\delta]-\Xi^{(4)}_{G}[\delta]=0$ on $\J_4$.

\subsection{Vanishing of the two-point function}

For $g=2,3$, $J^{(g)}=0$ identically on $\Sieg_g$, so that
eq.\eqref{twofinal2} simplifies to
\begin{equation*}\begin{split} A_2[\delta](a,b)= 
&128\hat A_2[\delta](a,b)+16\cdot 2\pi
i(1+\delta_{ij})\sum_{i,j}^g\omega_i(a)\omega_j(b)\\
&\qquad\qquad\qquad\qquad\cdot\Bigl(\frac{\partial\Xi^{(g)}[\delta](\Omega)}{\partial\Omega_{ij}}
-16\Xi^{(g)}[\delta](\Omega)\frac{\partial}{\partial\Omega_{ij}}\log\theta[\delta](\Omega,0)\Bigr) \ .
\end{split}\end{equation*} This reproduces (up to an irrelevant factor) the ansatz
\eqref{faketwo} for the two-point function, plus a correction. After
summing over the spin structures, by \eqref{sumdelta} the correction
vanishes, so that, by \eqref{faketwog3},
\be\label{twog3}\sum_{\delta\text{ even}}A_2[\delta](a,b)=0\ ,\qquad
g=2,3\ , \ee as expected from space-time supersymmetry. Note,
however, that the meaning of this result is quite different from the
analogous results obtained so far in the literature. In fact, here
we have made no further assumptions on the form of the two-point
function, beyond the ansatz for the chiral measure and the natural
factorisation properties of string amplitudes. Furthermore, the fact
that the two-point function vanishes at genus $g$ is really a check
for the chiral measure at genus $g+1$ rather than $g$. For such
reasons, \eqref{twog3} is a strong argument supporting the ansatz
for the chiral measure at genus three and four.

We now directly show that the two-point function at $g=4$, implied by the OPSMY ansatz for the measure, does not vanish at $g=4$. To this end, it is convenient to
choose $\Xi^{(4)}_G[\delta]$ on the second line of \eqref{twofinal2},
so that, after summing over the even spin structures, we can use
\eqref{sumdelta}
to simplify this expression (note that the coefficient of the polar part is proportional to the cosmological constant and vanishes)
$$A_2(a,b)=2^3(2^4+1)\bigl(-8 D_4+ 16 B_4+77 C-31 B_5 \bigr)\sum_{i,j}^g\omega_i(a)\omega_j(b)2\pi
i(1+\delta_{ij})\frac{\partial J^{(4)}}{\partial\Omega_{ij}}\ , $$
and being
$$ -8 D_4+ 16 B_4+77 C-31 B_5=-\frac{2^{14}}{7\cdot 11\cdot 17}\neq
0\ ,$$ we conclude that the two-point function obtained by factorisation from the OPSMY ansatz at $g=5$ does not vanish.

\section{Conclusions}

The renewed recent interest in trying to solve long
standing questions in superstring theory, mainly due to basic papers by D'Hoker and Phong, led to a parallel deeper analysis of the structure of moduli space of Riemann surfaces
involving Riemann theta functions, Siegel modular forms and theta series associated to unimodular lattices.

In the present paper we have seen that a careful use of such
mathematical results, combined with the old idea of factorisation of
string and conformal field theory amplitudes under degeneration
limits of Riemann surfaces, provide powerful tools to analyse the
structure of superstring amplitudes that would be inaccessible to a
direct calculation. A key point that simplifies considerably the
computations, concerns the choice of the local coordinate at the
node on degenerate Riemann surfaces. Our techniques lead to several
advantages with respect to other approaches to the problem, which
were based on strong assumptions about the form of these amplitudes.
On one hand, one can obtain information on the connected part of the
$n$-point function at a certain genus, once the chiral superstring
measure is known at higher loop. This could lead to a major advance
in the so far prohibitive task of computing higher loop $n$-point
functions in the RNS formulation of superstrings. On the other hand,
one can use the
 non-renormalisation theorems for one-, two- and three-point
functions to check consistency of the chiral measure at higher
genus, without introducing any further assumptions.

We applied this procedure to obtain a general expression the (spin
dependent part of the) chiral two-point function for two NS massless
states on a surface of genus $g$, for every $g$, from factorisation of the chiral
measure at genus $g+1$. Then, we specialised our result
to the recent ans\"atze for the chiral superstring measure and
explicitly compute the two-point function up to genus $4$. We proved
that, after GSO projection, the two-point function vanishes at
$g=2,3$ as expected from space-time supersymmetry and, in
particular, that the connected and the disconnected part of the
amplitude vanish separately.

We also showed that the same result does not hold for the genus four
two-point function obtained from the OPSMY ansatz for the chiral
measure at genus five. In this case, the connected and disconnected
part, after summing over the spin structures, give the same
non-vanishing contribution up to a factor, but these contributions
do not cancel each other. This probably means that OPSMY ansatz has
to be modified. For such a reason, it would be very interesting to
understand whether Grushevsky expression for the chiral measure is
equivalent to OPSMY at genus five. If they are different, we can
conjecture that a certain linear combination of the two ans\"atze
exists, leading to a vanishing two-point function at genus four. If
this is the case, then the vanishing of the two-point function at
genus $g$ should be imposed as an additional constraint for the
chiral measure at genus $g+1$. On the contrary, if Grushevsky and
OPSMY expressions are equivalent, it would be interesting to
understand whether they are the unique solutions to the constraints.

Another direction for further investigation concerns the computation
of the three-point functions at genus $g$ by multiple factorisation
of the chiral measure at genus $g+2$. In this respect, it is
interesting to observe that the disconnected part of the three-point
function vanish at genus $g=2$ but not at genus $g=3$
\cite{Matone:2008td}. Because these amplitudes can be obtained from
multiple factorisation of the chiral measure at genera $g+2=4,5$, it
is tempting to conjecture that this is related to the vanishing of
the disconnected part of the two-point function at genera $g+1=3,4$,
respectively. Finally, our techniques could be checked by computing
the four-point function at genus two and comparing it with the
results of \cite{D'Hoker:2005jc}. All such computations involve,
however, multiple degenerations limits and are technically more
complicated.

\section*{Acknowledgements}

The research of R.V. is supported by
an INFN Fellowship.

\appendix

\section{Theta functions and Riemann surfaces}\label{s:mathback}

Here we first provides some background on theta functions
and Riemann surfaces (see \cite{Fay:1973,Mumford:1983,Farkas:1992} for proofs and
details). Next, we consider the degeneration of Riemann surfaces which is used in section three to derive the two-point function. We also
derive a basic formula for a section of $|2\Theta|$, with
$\Theta$ denoting the theta divisor.

\subsection{Definitions and basic results}\label{appendiceuno}

Let $\Sieg_g$ denote the Siegel upper half-space, i.e. the space of
$g\times g$ complex symmetric matrices with positive definite
imaginary part
$$\Sieg_g:=\{\Omega\in M_{g\times g}(\CC)\mid \tp\Omega=\Omega\,
,\im\Omega>0 \}\ .
$$
Let $\Sp(2g,\ZZ)$ be the symplectic modular group, i.e. the group of
$2g\times 2g$ complex matrices $M:=\bigl(\begin{smallmatrix}A &B\\
C & D\end{smallmatrix}\bigr)$, where $A,B,C,D$ are $g\times g$
blocks satisfying
$$\tp A C=\tp C A\ ,\quad \tp BD=\tp D B\ ,\quad \tp
D A-\tp B C=\II_g\ .
$$
Let us define the action of $\Sp(2g,\ZZ)$ on $\CC^g\times\Sieg_g$ by
\begin{equation}\label{modull}
(M\cdot z,M\cdot\Omega):=\bigl(\tp(C\Omega+D)^{-1}z,(A\Omega+B)(C\Omega+D)^{-1}\bigl)\ ,\end{equation} where $M:=\bigl(\begin{smallmatrix}A &B\\
C & D\end{smallmatrix}\bigr)\in\Sp(2g,\ZZ)$ and
$(z,\Omega)\in\CC^g\times\Sieg_g$.

 For each $\delta',\delta'' \in\FF_2^{g}$, the theta
function $\theta[\delta]:=\theta\thspin{\delta'}{\delta''}\colon
\CC^g\times\Sieg_g\to \CC$ with characteristics $[\delta]
:=\thspin{\delta'}{\delta''}$
 is defined by
$$\theta[\delta](z,\Omega):=\sum_{k\in\ZZ^g}
\exp{\pi
i\Bigl[\tp\Bigl(k+\frac{\delta'}{2}\Bigr)\Omega\Bigl(k+\frac{\delta'}{2}\Bigr)+
2\tp\Bigl(k+\frac{\delta'}{2}\Bigr)\Bigl(z+\frac{\delta''}{2}
\Bigr)\Bigr]\ ,}
$$
where $(z,\Omega)\in\CC^g\times\Sieg_g$. For each fixed $\Omega$,
$\theta[\delta](z,\Omega)$ is an even or odd function on $\CC^g$
depending whether $(-1)^{\delta'\cdot\delta''}$ is $+1$ or $-1$,
respectively. Correspondingly, there are $2^{g-1}(2^g+1)$ even and
$2^{g-1}(2^g-1)$ odd theta characteristics. Under translations
$z\mapsto z+\lambda$, $z\in\CC^g$, $\lambda
\in\ZZ^g+\Omega\ZZ^g\subset\CC^g$, theta functions get multiplied by
a nowhere vanishing factor
$$\theta \thspin{\delta'}{\delta''}\left(z+n+\Omega m,\Omega\right)=
e^{-\pi i \tp{m}\Omega m-2\pi i\tp{m}z+\pi
i(\tp{\delta'}n-\tp{\delta''}m)}\theta\thspin{\delta'}{\delta''}\left(z,\Omega\right)
\ ,
$$
$m,n\in\ZZ^g$. It follows that, for any fixed $\Omega$, the theta
functions can be seen as sections of line bundles on the complex
torus $A_\Omega:=\CC^g/(\ZZ^g+\Omega\ZZ^g)$, with a well defined
divisor on $A_\Omega$. We denote by $\Theta$ the divisor of
$\theta(z)=\theta[0](z,\Omega)=\theta\thspin{0}{0}(z,\Omega)$.

\medskip

The action of $\Sp(2g,\ZZ)$ on the space of theta characteristics $\FF_2^{2g}$ is defined by
 \begin{equation}
\label{modulcharac}[\delta\cdot M]=
\biggl[\begin{pmatrix}\delta' \\
\delta''\end{pmatrix}\cdot M\biggr]:=\begin{pmatrix}\tp A & \tp C\\ \tp B
&\tp D\end{pmatrix}\begin{bmatrix}\delta' \\
\delta''\end{bmatrix}+\begin{bmatrix} (\tp A
C)_0\\ (\tp BD)_0\end{bmatrix}\mod 2\ ,
\end{equation} where, for any matrix $A$, we denote by $A_0$ the vector of diagonal entries.
Theta characteristics are invariant under the action of the subgroup
$\Gamma(2)\subset\Sp(2g,\ZZ)$, where
$$\Gamma(n):=\{M\in\Sp(2g,\ZZ)\mid M=\II_{2g}\mod n\}\ ,
$$ is the subgroup of elements of $\Sp(2g,\ZZ)$ congruent to the $2g\times 2g$
identity matrix mod $n$. The theta characteristic $[0]:= \spin{0}{0}$ is fixed by the subgroup
$$\Gamma(1,2):=\{\left(\begin{smallmatrix}A & B\\ C& D\end{smallmatrix}\right)\in \Sp(2g,\ZZ)\mid
(\tp A C)_0\equiv (\tp BD)_0\equiv 0\bmod 2\}\ .$$
Symplectic
transformations preserve the parity of the characteristics and, for
any two $\delta,\epsilon\in\ZZ_2^{2g}$ of the same parity, there
exists $M\in\Sp(2g,\ZZ_2)$ such that $\epsilon=M\cdot\delta$.

A (Siegel) modular form $f$ of weight $k\in\ZZ$ for a subgroup
$\Gamma\in\Sp(2g,\ZZ)$ is a holomorphic function on $\Sieg_g$ such, for all $M\in\Gamma$,
that
$$f(M\cdot \Omega)=\det(C\Omega+D)^kf(\Omega)\ .
$$
A condition of regularity, automatically satisfied for $g>1$, is also required for
$g=1$.

Let $C$ be a Riemann surface of genus $g>1$. The choice of a marking
for $C$ provides a set of generators
$\{\alpha_1,\ldots,\alpha_g,\beta_1,\ldots,\beta_g\}$ for the first
homology group $H_1(C,\ZZ)$ on $C$, with symplectic intersection
matrix, that is
\begin{equation}\label{sympl}\alpha_i\cdot\alpha_j=0=\beta_i\cdot\beta_j\ ,\qquad
\alpha_i\cdot\beta_j=\delta_{ij}\ ,
\end{equation}
for all $i,j=1,\ldots,g$. The choice of such generators canonically
determines a basis $\{\omega_1,\ldots,\omega_g\}$ for the space
$H^0(K_C)$ of holomorphic $1$-differentials on $C$, with normalised
$\alpha$-periods
$\oint_{\alpha_i}\omega_j=\delta_{ij}$,
for all $i,j=1,\ldots,g$. The $\beta$-periods define the Riemann
period matrix
$\Omega_{ij}:=\oint_{\beta_i}\omega_j$,
which is symmetric and with positive-definite imaginary part, so
that $\Omega\in\Sieg_g$. By Torelli's theorem, the complex structure
of $C$ is completely determined by its Riemann period matrix.

The conditions \eqref{sympl} determine the basis
of $H_1(C,\ZZ)$ up to a symplectic transformation
$$ 
\begin{pmatrix}\alpha\\
\beta\end{pmatrix}\mapsto 
\begin{pmatrix}\tilde\alpha\\
\tilde\beta\end{pmatrix}:=\begin{pmatrix}D & C \\ B &
A\end{pmatrix} \begin{pmatrix}\alpha\\
\beta\end{pmatrix}\ ,\qquad\qquad
M=\begin{pmatrix}A & B
\\ C & D\end{pmatrix}\in \Sp(2g,\ZZ) \ , $$ under which
$
(\omega_1,\ldots,\omega_g)\mapsto(\tilde\omega_1,\ldots,\tilde\omega_g):=(\omega_1,\ldots,\omega_g)(C\Omega+D)^{-1}$,
 whereas $\Omega\mapsto\tilde\Omega:=M\cdot\Omega$
transforms as in \eqref{modull}.

The complex torus $A_C:=\CC^g/(\ZZ^g+\Omega\ZZ^g)$ associated to the
Riemann period matrix of $C$ is called the Jacobian torus of $C$.
For a fixed base-point $p_0\in C$, let $I\colon C\to A_C$ denote the
Abel-Jacobi map, defined by
$$p\mapsto
I(p):=\tp\Bigl(\int_{p_0}^p\omega_1,\ldots,\int_{p_0}^p\omega_g\Bigr)\in
A_C\ .
$$
Note that different choices of the path of integration from $p_0$ to
$p$ correspond, by the formula above, to points in $\CC^g$ differing
by elements in the lattice $\ZZ^g+\Omega\ZZ^g$, so that $I$ is
well-defined only on $\CC^g/(\ZZ^g+\Omega\ZZ^g)$. The Abel-Jacobi
map extends to a map from the Abelian group of divisors on $C$ to
$A_C$ by
$$I\bigl(\smsum\nolimits_i p_i-\smsum\nolimits_i q_i\bigl):=\smsum\nolimits_i I(p_i)-\smsum\nolimits_i I(q_i)\ .$$
Such a map is independent of the base-point $p_0$ when restricted to
zero degree divisors. When no confusion is possible, we will
identify such zero degree divisors with their image in $A_C$ through
$I$. In particular, we will omit both $I$ when considering the theta
functions on the Jacobian evaluated at (the image of) some zero
degree divisor on $C$ and the argument $\Omega$ for theta
functions associated to a marked Riemann surface.

Fix a non-singular odd theta characteristic $\nu\in\FF_2^{2g}$ and
consider $\sum_{i=1}^g\partial_i\theta[\nu](0)\omega_i$,
which is a holomorphic $1$-differential with $g-1$ double
zeroes and is the square $h_\nu^2$ of a holomorphic
$1/2$-differential with odd spin structure $\nu$. This differential defines
the prime form
\be 
E(a,b):=\frac{\theta[\nu](b-a)}{h_\nu(a)h_\nu(b)}\ ,
\ee $a,b\in C$, which is a section of a line bundle on $C\times C$,
antisymmetric in its arguments, vanishing only on the
diagonal $a=b$ and independent of the choice of
$\nu$.

For each non-singular even characteristic $\delta\in\FF_2^{2g}$, the
Szeg\"o kernel is the meromorphic $1/2$-differential
$$S_\delta(a,b):=\frac{\theta[\delta](a-b)}
{\theta[\delta](0)\,E(a,b)} \ ,
$$
with a single pole at $a=b$ and
holomorphic elsewhere.

Finally, we denote by
\be\label{omegathird}\omega_{a-b}(x):=\frac{d}{d
x}\log\frac{E(x,a)}{E(x,b)}\ ,\ee $a,b,x\in C$, the Abelian
$1$-differential of the third kind with single poles on $a$ and $b$
with residue $+1$ and $-1$, respectively, holomorphic elsewhere and
with vanishing $\alpha$-periods, and with
\be\label{omegasecond}\omega(a,b):=\frac{d^2}{d
a\, db}\log E(a,b)\ ,\ee the Abelian $1$-differential of the second
kind with a double pole of residue $1$ at $a=b$, holomorphic
elsewhere and with vanishing $\alpha$-periods.

\subsection{Degeneration formulae}\label{a:degen}

Here we derive the degeneration formulae for the
Riemann period matrix. A key point concerns the local coordinate at the node of degenerate Riemann
surfaces whose choice leads to a considerable simplification of the calculations to derive the two-point function from the chiral measure.

Consider two distinct points
$p_1,p_2\in C$ and let $z_1,z_2$ be local coordinates
$$ z_{i}:\{z\in\CC\mid |z|<1\}\to U_{i}\subset C\ , \qquad z_i(p_i)=0\
,\qquad\qquad i=1,2\ ,$$ centered at $p_1$ and
$p_2$, respectively.
 Then, a family
$$ \{\tilde C_t\mid t\in\CC, 0<|t|<1\} \ ,
$$ of Riemann surfaces of genus $g+1$ is defined, where
$$ \tilde C_t=C\setminus (U_{1,t}\cup U_{2,t})\ ,
$$ with $U_{i,t}:=\{p\in U_i\mid |z_i(p)|\le |t|\}$, $i=1,2$, and two
points $p\in U_1\setminus U_{1,t}$ and $q\in U_2\setminus U_{2,t}$
are identified if
$$ z_1(p)z_2(q)=t\ .
$$ Let $\alpha_1,\ldots,\alpha_g,\beta_1,\ldots,\beta_g$ be a symplectic basis for the homology
of $C$, with representatives in $C\setminus (U_1\cup U_2)$ and $\omega_1,\ldots,\omega_g$ the basis of canonically normalised abelian differentials.
We can
choose a basis
$\tilde\alpha_1(t),\ldots,\tilde\alpha_{g+1}(t),\tilde\beta_1(t),\ldots,\tilde\beta_{g+1}(t)$
of $H_1(\tilde C_t,\ZZ)$ such that
$$\tilde\alpha_i(t)=\alpha_i\ , \qquad \tilde\beta_i(t)=\beta_i\ ,\qquad
i=1,\ldots, g\ . $$ As representatives of $\tilde\alpha_{g+1}(t)$ and
$\tilde\beta_{g+1}(t)$ we can consider, respectively, the circle
$|z_1|=\sqrt{|t|}$ and a suitable path on $C$ from $z_1^{-1}(x)$ to
 $z^{-1}_2(t/x)$ for some $x\in\CC$, $|t|<|x|<1$. Then, the Riemann period matrix
$\tilde\Omega(t)$ of $\tilde C_t$ with respect to this basis is
\cite{Fay:1973,Yamada:1980}
$$ \tilde\Omega(t)=\begin{pmatrix}
\Omega_{ij}+2\pi it\sigma_{ij} & \int_{p_1}^{p_2}\omega_i+t\sigma_i\\
\int_{p_1}^{p_2}\omega_j+t\sigma_j & \frac{1}{2\pi i}\log t
+c_0+c_1t
\end{pmatrix}+O(t^2)\ ,$$
where
\begin{align}
\sigma_{ij}&=-\frac{\omega_i(p_1)\omega_j(p_2)+\omega_i(p_2)\omega_j(p_1)}{dz_1(p_1)\,dz_2(p_2)}\ , &
\sigma_{i}&=-\Bigl(\gamma_1\frac{\omega_i}{dz_2}(p_2)+\gamma_2\frac{\omega_i}{dz_1}(p_1)\Bigr)\
,\\[5pt]
\label{czerouno} c_0&=\frac{1}{2\pi i}\lim_{x\to
0}\Bigl(\int_{z_1^{-1}(x)}^{z_2^{-1}(x)}\omega_{p_2-p_1}-2\log
x\Bigr)\ ,& c_1&=\frac{i}{\pi}\gamma_1\gamma_2\ ,
\end{align}
where $\omega_{p_2-p_1}$ is the $1$-differential of the third kind on $C$
defined in \eqref{omegathird}, and
\be\label{gammai} \gamma_i=\lim_{x\to
p_i}\Bigl(\omega_{p_1-p_2}(x)-(-1)^i\frac{dz_i}{z_i}(x)\Bigr)\
,\qquad i=1,2\ .\ee All these parameters can be exactly computed for
a suitable choice of the coordinates $z_1$ and $z_2$. Choose $2g$
curves on $C$ which are representatives of the basis of homology and
consider the canonical dissection of $C$ along these curves. We can
identify $C$ with a fundamental domain $\hat C$ in the upper
half-plane $\mathbb{H}$ with respect to the Fuchsian uniformisation
on $C$. Let us choose such a dissection so that $p_1,p_2$ and the
paths $\tilde\alpha_{g+1}$ and $\tilde\beta_{g+1}$ lie in the interior of
$\hat C$. Fix an arbitrary point $c\in \hat C$, distinct from
$p_1,p_2$, and set
$$ z_1(p):=\frac{E(p,p_1)E(c,p_2)}{E(p,p_2)E(c,p_1)}=
e^{\int_c^p\omega_{p_1-p_2}}\ ,\qquad
z_2(q):=\frac{E(q,p_2)E(c,p_1)}{E(q,p_1)E(c,p_2)}=
e^{\int_c^q\omega_{p_2-p_1}}\ . $$ These coordinates, that
represent the higher genus generalisations of cross-ratios on the
sphere, satisfy the following
properties
\begin{align*}dz_1(p)&=\omega_{p_1-p_2}(p)z_1(p)\ ,&\qquad
dz_1(p_1)&=\frac{E(c,p_2)}{E(p_2,p_1)E(c,p_1)}\
,\\dz_2(q)&=\omega_{p_2-p_1}(q)z_2(q)\ ,&\qquad
dz_2(p_2)&=\frac{E(c,p_1)}{E(p_1,p_2)E(c,p_2)}\ ,\end{align*}
where $p,q$ are distinct from $p_1,p_2$.
Replacing these expressions in \eqref{gammai}, it follows
immediately that $\gamma_i=0$ in such coordinates. Furthermore, if
in \eqref{czerouno} we choose a path from $z_1^{-1}(x)$ to
$z_2^{-1}(x)$ in $\hat C$ passing through $c$, we obtain
$$
\int_{z_1^{-1}(x)}^{z_2^{-1}(x)}\omega_{p_2-p_1}=\int^{z_1^{-1}(x)}_{c}\omega_{p_1-p_2}+
\int_{c}^{z_2^{-1}(x)}\omega_{p_2-p_1}=\int_1^x
\frac{dz_1}{z_1}+\int_1^x \frac{dz_2}{z_2}=2\log x\ ,$$ where we
used $z_1(c)=1=z_2(c)$. It follows that $c_0=0$ and
 and we finally obtain
$$\tilde\Omega(t)=\begin{pmatrix}
\Omega_{ij}+2\pi it\sigma_{ij} & \int_{p_1}^{p_2}\omega_i\\
\int_{p_1}^{p_2}\omega_j & \frac{1}{2\pi i}\log t
\end{pmatrix}+O(t^2)\ ,$$
where
$$
\sigma_{ij}=-\frac{\omega_i(p_1)\omega_j(p_2)+\omega_i(p_2)\omega_j(p_1)}{dz_1(p_1)\,dz_2(p_2)}=
E(p_1,p_2)^2(\omega_i(p_1)\omega_j(p_2)+\omega_i(p_2)\omega_j(p_1))\
. $$

\subsection{A formula for the sections of $|2\Theta|$}\label{a:twotheta}

Fix an element $\Omega\in\Sieg_g$ an consider the complex torus
$A_\Omega:=\CC^g/(\ZZ^g+\Omega\ZZ^g)$. A section of $|k\Theta|$ on $A_\Omega$, $k\in\NN$,
corresponds to a holomorphic function $F$ on $\CC^g$ obeying the
quasi-periodicity conditions
$$F(z+\Omega m+n)=e^{-k(\pi i \tp{m}\Omega m+2\pi i\tp{m}z)}F(z)\
,\qquad m,n\in\ZZ^g\ .$$ In the following we will be interested in
the space $H^0(A_\Omega,|2\Theta|)$ of sections of $|2\Theta|$,
which is spanned by the squares $\theta[\delta]^2(z)$ of theta
functions with characteristics. In particular, all such sections are
even functions of $z$. Let $\Omega$ be the period matrix of a
Riemann surface $C$ and $A_C$ its Jacobian, and consider the
restriction of a section $F\in H^0(A_C,|2\Theta|)$  to the locus
$$\textstyle{C-C:=\{(b-a):= (\int_a^b\omega_1,\ldots,\int_a^b\omega_g)\mid a,b\in C\}\subseteq A_C\
.} $$ It is easy to see that
$$
\frac{F(b-a)}{E(a,b)^2}=\frac{F(b-a)}{\theta[\nu]^2(b-a)}h_\nu^2(a)h_\nu^2(b)\
, $$ is a single-valued meromorphic $1$-differential with respect to
$(a,b)\in C\times C$, with a double pole on the diagonal $a=b$ and
holomorphic elsewhere. The space of such differentials is generated
by $\{\omega_i(a)\omega_j(b)\}_{i,j=1,\ldots,g}$ and by the
normalised differential of the second kind $\omega(a,b)$ defined in
\eqref{omegasecond}, so that \be\label{duetheta}
F(b-a)=E(a,b)^2\Bigl(c_0\omega(a,b)+\sum_{i,j}^gc_{ij}\omega_i(a)\omega_j(b)\Bigr)\
. \ee To compute the coefficients $c_0,c_{ij}$, let us compare the
expansion of both sides of \eqref{duetheta} in the limit of $b\to
a$. Since $F$ is even, we have
$$F(b-a)=F(0)+\frac{1}{2}(b-a)^2\sum_{i,j}^g\partial_i\partial_j
F(0)\omega_i(a)\omega_j(a)+O(b-a)^4\ .$$  On the other hand
\cite{Fay:1973},
\begin{align*}da^{1/2}\,db^{1/2}\, E(a,b)&=(b-a)-\frac{1}{12}S(a)(b-a)^3+O(b-a)^5\ ,\\
\omega(a,b)&=da\,db((b-a)^{-2}+\frac{1}{6}S(a)+O(b-a)^2)\ ,
\end{align*} with $S(a)$ a holomorphic projective connection, so that the right hand side of \eqref{duetheta} becomes
$$c_0+(b-a)^2\sum_{i,j}^gc_{ij}\omega_i(a)\omega_j(a)+O(b-a)^4\ .
$$ It follows that
$$ F(b-a)=E(a,b)^2\Bigl(F(0)\omega(a,b)+\frac{1}{2}\sum_{i,j}^g\partial_i\partial_j
F(0)\omega_i(a)\omega_j(b)\Bigr)\ .$$ Note that, since
$\theta[\delta]^2(z)\in H^0(A_C,|2\Theta|)$, this relation implies
eq.\eqref{questa}.

\section{Theta series and lattices}\label{a:thetaseries}

We collect here some useful results about unimodular
lattices and their theta series.

\subsection{Proof of formula \eqref{thserexpr}}

To each even theta characteristic $\delta:=\thspin{\delta'}{\delta''}\in\FF_2^{2g}$ associate an element $M_\delta\in \Sp(2g,\ZZ)$ such that
\be\label{modulzero} [0\cdot M_\delta]:=\begin{bmatrix}
(\tp AC)_0\\ (\tp B D)_0
\end{bmatrix}=[\delta]\ .\ee
In particular, we can choose
\be\label{emmedelta} M_\delta=\begin{pmatrix}
\diag(\delta') & -\II_g\\ \II_g & 0
\end{pmatrix}\begin{pmatrix}
\II_g & S\\ 0 & \II_g
\end{pmatrix}=\begin{pmatrix}
\diag(\delta') & \diag(\delta')S-\II_g\\ \II_g & S
\end{pmatrix}\ , \ee
where, for any vector $v=(v_1,\ldots,v_g)$, and $\diag(v)$ denotes the diagonal matrix with $\diag(v)_{ii}=v_i$.
Here, $S$ is an integer $g\times g$ matrix satisfying
$$ \delta''=S\delta'+S_0\ ,\qquad S_0\cdot \delta'=0\ .
$$ For example, if
$$ \spin{\delta'}{\delta''}= 
\spin{1 & 1 & 1 & 1 & 1 & 1 & 1 & 0 & 0 & 0}{
1 & 1 & 1 & 1 & 0 & 0 & 0 & 1 & 1 & 0}
\ ,$$
then
$$ S=\left(\begin{smallmatrix}
0 & 1 & \\
1 & 0 & \\
 && 0 & 1 & \\
 && 1 & 0 &\\
 &&&& 0\\
 &&&&& 0\\
 &&&&&& 0\\
 &&&&&&& 1\\
 &&&&&&&& 1\\
 &&&&&&&&& 0
\end{smallmatrix}\right)\ , $$ and this construction can be easily generalised
to every even $\delta$. Note that
$$
M_\delta\cdot\Omega=(\diag(\delta')(\Omega+S)-\II)(\Omega+S)^{-1}=\diag(\delta')+\hat\Omega\
, $$ where
$$ \hat\Omega=-(\Omega+S)^{-1}\ .
$$
Let $\Lambda$ be a $d$-dimensional unimodular lattice, with $d\equiv 0\mod 8$.
Choose a basis $\lambda^{(1)},\ldots,\lambda^{(d)}\in\RR^{d}$ of generators of $\Lambda$ and let $E$ be
the $d\times d$ matrix whose $i$-th column is the vector $\lambda^{(i)}$, $i=1,\ldots,d$, and $Q$ the Gram matrix
$$ E:=(\lambda^{(1)},\ldots,\lambda^{(d)})\ ,\qquad\qquad Q_{ij}:=\lambda^{(i)}\cdot\lambda^{(j)}\ .
$$
Then, by construction, $Q=\tp EE$ is an integral unimodular matrix.
Let $\Theta_{\Lambda}$ be the theta series \eqref{thetaseries} of $\Lambda$ and, as in \eqref{defseries}, set
$$
\Theta_\Lambda[\delta](\Omega)=(\det\hat\Omega)^{d/2}\Theta_\Lambda(\diag(\delta')+\hat\Omega)\
. $$ Let us define $r''\in(\ZZ/2\ZZ)^{d}$ and $u\in\Lambda$ by
$$ r''\equiv Q_0\bmod 2\ , \qquad\qquad u:=EQ^{-1}r''= \tp E^{-1} r''\ ,
$$ and notice that for any $\lambda=En\in\Lambda$, $n\in\ZZ^{d}$,
we have
$$ \lambda\cdot\lambda=\tp n Q n=\sum_i^d
n_i^2Q_{ii}+2\sum_{i<j}^dn_iQ_{ij}n_j\equiv\sum_i^d n_iQ_{ii}\mod 2\ , $$
so that \be\label{parityvec} \lambda\cdot\lambda \equiv n\cdot
r''\equiv \lambda\cdot u\mod 2\ ,\qquad \text{for all }\lambda\in\Lambda\ .\ee A vector $u\in\Lambda$ satisfying this property is called a parity (or characteristic) vector for $\Lambda$. Thus $$
\Theta_\Lambda(\diag(\delta')+\hat\Omega)=\sum_{\lambda_1,\ldots,\lambda_g\in\Lambda}e^{\pi
i\sum_{i,j}\lambda_i\cdot\lambda_j\hat\Omega_{ij}+\pi
i\sum_i\delta'_i\lambda_i\cdot\lambda_i}=\sum_{\lambda_1,\ldots,\lambda_g\in\Lambda}e^{\pi
i\sum_{i,j}\lambda_i\cdot\lambda_j\hat\Omega_{ij}+\pi
i\sum_i\delta'_i\lambda_i\cdot u}\ ,$$ that can be rewritten as
$$ \Theta_\Lambda(\diag(\delta')+\hat\Omega)=\sum_{n_1,\ldots,n_g\in\ZZ^{d}}e^{\pi
i\sum_{i,j}(\tp n_i Qn_j)\hat\Omega_{ij}+2\pi
i\sum_i\frac{\delta'_i}{2} n_i\cdot r''}\ ,$$ so that making a Poisson
resummation with respect to $(n_1,\ldots,n_g)\in\ZZ^{g\times d}$ becomes
$$ \Theta_\Lambda(\diag(\delta')+\hat\Omega)=(\det\hat\Omega)^{-d/2}
\sum_{m_1,\ldots,m_g\in\ZZ^{d}}e^{\pi i\sum_{i,j}\tp
(m_i+\frac{\delta'_i}{2}r'')Q^{-1}(m_j+\frac{\delta'_j}{2}r'')(\Omega_{ij}+S_{ij})}\
. $$ Set
$$ n_i=Q^{-1}m_i\ ,\qquad\qquad r'=Q^{-1}r''\ ,
$$ and use $Q^{-1}\in  {\rm GL}(d,\ZZ)$ and $u=Er'$ to
obtain
\begin{equation*}\begin{split} \Theta_\Lambda[\delta](\Omega)=&\sum_{n_1,\ldots,n_g\in\ZZ^{d}}e^{\pi
i\sum_{i,j}\tp(n_i+\frac{\delta'_i}{2}r')Q(n_j+\frac{\delta'_j}{2}r')(\Omega_{ij}+S_{ij})}\\=&
\sum_{\lambda_1,\ldots,\lambda_g\in\Lambda}e^{\pi
i\sum_{i,j}(\lambda_i+\frac{\delta'_i}{2}u)\cdot(\lambda_j+\frac{\delta'_j}{2}u)(\Omega_{ij}+S_{ij})}\
. \end{split}\end{equation*} Observe that
$$ \sum_{i,j}\lambda_i\cdot\lambda_jS_{ij}\equiv \sum_i
\lambda_i\cdot\lambda_iS_{ii}\equiv  \sum_i \lambda_i\cdot
uS_{ii}\mod 2\  ,$$ because $S$ is integral and symmetric.
Furthermore, any parity vector $u$ satisfies \cite{Elkies1995}
$$ u\cdot u\equiv d\bmod 8\equiv 0\bmod 8\ ,
$$ (for $d=16$ this can also be checked by an explicit case by case calculation), so
that the following congruences mod $2$ hold
\begin{equation*}\begin{split}
\sum_{i,j}&(\lambda_i+\frac{\delta'_i}{2}u)\cdot(\lambda_j+\frac{\delta'_j}{2}u)S_{ij}\equiv
\sum_i \lambda_i\cdot u((S\delta')_i+(S_0)_i)+\sum_{i,j}\frac{u\cdot
u}{4}\delta'_iS_{ij}\delta'_j\\
&\equiv \sum_i \lambda_i\cdot u\delta''_i\equiv \sum_i
(\lambda_i+\frac{\delta'_i}{2}u)\cdot u\delta''_i\mod
2\ , \end{split}\end{equation*} and \eqref{thserexpr} follows.
Note that the set of parity vectors, i.e. the
vectors in $\Lambda$ satisfying \eqref{parityvec}, is given by
$u+2\Lambda$,  and \eqref{thserexpr} does not change if we replace
$u$ by an arbitrary $\tilde u\in u+2\Lambda$. Also note that $\tilde
u\cdot\tilde u\equiv u\cdot u\mod 8$ (this property holds for
unimodular lattices of any dimension \cite{Elkies1995}). The definition
\eqref{thserexpr} makes sense also for $\delta$ an odd theta
characteristic, but in this case there is no $M\in \Sp(2g,\ZZ)$ satisfying \eqref{modulzero}.

\subsection{Sums over spin structures}

Let $\Lambda$ be an odd unimodular lattice and $\Lambda_e\subset\Lambda$ the sublattice of vectors of
even norm, so that
$\Lambda_e\subset\Lambda\subset\Lambda_e^*$. If $u\in\Lambda$ is a
parity vector, i.e. satisfies \eqref{parityvec}, then
$u/2\in\Lambda_e^*$ and it maps to a non-trivial element of
$\Lambda_e^*/\Lambda\cong\ZZ_2$. Let $\lambda_o\in\Lambda$ be an
arbitrary vector of odd norm and $\Lambda_{o}=\lambda_o+\Lambda_e$
the set of vectors of odd norm, so that
$\Lambda=\Lambda_e\cup\Lambda_o$. We have the decomposition
$$
\Lambda_e^*=\Lambda\cup(\frac{u}{2}+\Lambda)=
\Lambda_e\cup\Lambda_o\cup(\frac{u}{2}+\Lambda_e)\cup(\frac{u}{2}+\Lambda_o)\
. $$ Set \be\label{la1la2}
\Lambda^{(1)}:=\Lambda_e\cup(\frac{u}{2}+\Lambda_e)\ ,\qquad
\Lambda^{(2)}:=\Lambda_e\cup(\frac{u}{2}+\Lambda_o)\ .\ee
\begin{proposition}
If $\Lambda$ is a $d$-dimensional unimodular lattice, with $d\equiv
0\mod 8$, then $\Lambda^{(1)}$ and $\Lambda^{(2)}$ are
$d$-dimensional even unimodular lattices.
\end{proposition}
\begin{proof}
For $d$-dimensional unimodular lattices, the norm of a parity vector
satisfies $u\cdot u\equiv d\mod 8$ \cite{Elkies1995}. It follows that, for
$d\equiv0\bmod 8$,
$$ \frac{u}{2}\cdot \frac{u}{2}=\frac{u\cdot u}{4}\in 2\ZZ\ ,\qquad
(\frac{u}{2}+\lambda_o)\cdot (\frac{u}{2}+\lambda_o)=\frac{u\cdot
u}{4}+u\cdot \lambda_o+\lambda_o\cdot\lambda_o\in 2\ZZ\ .$$ In
particular, $u$ and $u+2\lambda_o$ are elements of $\Lambda_e$, so
that $\Lambda^{(1)}$ and $\Lambda^{(2)}$ are closed under the sum.
Furthermore, they are integral
($\Lambda^{(i)}\subseteq{\Lambda^{(i)}}^*$) and even. To prove that
they are self-dual, first observe that
${\Lambda^{(i)}}^*\subset\Lambda_e^*$ because
$\Lambda_e\subset\Lambda^{(i)}$. Since
$$ \lambda_o\cdot\frac{u}{2}\ ,\quad\quad
\lambda_o\cdot(\frac{u}{2}+\lambda_o)\ ,\quad\quad
\frac{u}{2}\cdot(\frac{u}{2}+\lambda_o)\in \frac{1}{2}+\ZZ\ ,$$ we
conclude that ${\Lambda^{(i)}}^*\cap(\Lambda_e^*\setminus
\Lambda^{(i)})$ is empty.
\end{proof}

\noindent In particular, when $d=16$, $\Lambda^{(1)}$ and $\Lambda^{(2)}$ must
be isomorphic to either $D_{16}^+$ or $E_8^2$. The even lattices
corresponding to each $\Lambda_k$ can be found by considering its
set of vectors of norm $2$, which is the root system of the Lie
algebras $\g_k$ (see table \ref{t:lattices}). This root system must
be contained in $\Lambda^{(i)}_k$ and in most cases, namely for
$k>0$, there is only one even unimodular lattice
satisfying this constraint, so that
$\Lambda^{(1)}_k\cong\Lambda^{(2)}_k$. The only
exception is $\Lambda_0$, because $\g_0\cong D_8\oplus D_8$ can be
embedded both in $E_8\oplus E_8$ or in $D_{16}$. In this case, a
more detailed analysis of the root systems shows that
$\Lambda^{(1)}=\Lambda_{D_{16}^+}$ and
$\Lambda^{(2)}=\Lambda_{E_8^2}$.

\begin{proposition}\label{p:seriessum}
Let $\Lambda$ be a $d$-dimensional unimodular lattice, with $d\equiv 0\bmod 8$,
$\Lambda^{(1)}$ and $\Lambda^{(2)}$ defined as in \eqref{la1la2} and
$\Theta^{(g)}_\Lambda[\delta]$ the theta series \eqref{thserexpr}
for every (even or odd) theta characteristic $\delta$. Then
$$ \sum_{\delta\text{ even}}\Theta^{(g)}_\Lambda[\delta]=2^{g-1}(\Theta^{(g)}_{\Lambda^{(1)}}+\Theta^{(g)}_{\Lambda^{(2)}})\
, $$ and
$$ \sum_{\delta\text{
odd}}\Theta_\Lambda^{(g)}[\delta]=2^{g-1}(\Theta^{(g)}_{\Lambda^{(1)}}-\Theta^{(g)}_{\Lambda^{(2)}})
\ .
$$
\end{proposition}
\begin{proof} For any $0\le k\le g$,
$\lambda_{k+1},\ldots,\lambda_g\in\RR^{d}$, $\Omega\in\Sieg_g$ and
$2k$-dimensional theta characteristic $[\delta]=\thspin{\delta'}{\delta''}$,
$\delta',\delta''\in\FF_2^k$, let us define
$$
R^{(g)}_k[\delta](\lambda_{k+1},\ldots,\lambda_g,\Omega):=\sum_{\substack{\lambda_1,\ldots,\lambda_k\\
\lambda_i\in\Lambda+\delta_i'\frac{u}{2}}}(-1)^{\sum_{i=1}^k\delta_i''\lambda_i\cdot
u}\,e^{\pi i \sum_{i,j}^g\lambda_i\cdot\lambda_j \Omega_{ij}}\ ,\qquad
k>0\ ,$$ and $R^{(g)}_0(\lambda_1,\ldots,\lambda_g,\Omega):=e^{\pi
i \sum_{i,j}^g\lambda_i\cdot\lambda_j \Omega_{ij}}$. We will prove
that, for all $1\le k\le g$
\be\label{tobeproved}\begin{split} &\sum_{\delta\in\FF_2^{2k}\text{ even}}R^{(g)}_k[\delta]
=2^{k-1}\Bigl(\sum_{\lambda_1,\ldots,\lambda_k\in\Lambda^{(1)}}
R^{(g)}_0+
\sum_{\lambda_1,\ldots,\lambda_k\in\Lambda^{(2)}}R^{(g)}_0\Bigr)\
,\\ &\sum_{\delta\in\FF_2^{2k}\text{ odd}}R^{(g)}_k[\delta]
=2^{k-1}\Bigl(\sum_{\lambda_1,\ldots,\lambda_k\in\Lambda^{(1)}}
R^{(g)}_0-
\sum_{\lambda_1,\ldots,\lambda_k\in\Lambda^{(2)}}R^{(g)}_0\Bigr)\ .
\end{split}\ee  The proposition corresponds to the particular
case $k=g$. For all $1\le k\le g$ and
$[\hat\delta]=\thspin{\hat\delta'}{\hat\delta''}\in\FF_2^{2(k-1)}$,
we have
 \begin{align*}
R^{(g)}_k\thspin{\hat\delta'&0}{\hat\delta''&0}=&\sum_{\lambda_k\in\Lambda_e}R^{(g)}_{k-1}[\hat\delta]
+\sum_{\lambda_k\in\Lambda_o}R^{(g)}_{k-1}[\hat\delta]\ ,\\
R^{(g)}_k\thspin{\hat\delta'&0}{\hat\delta''&1}=&\sum_{\lambda_k\in\Lambda_e}R^{(g)}_{k-1}[\hat\delta]
-\sum_{\lambda_k\in\Lambda_o}R^{(g)}_{k-1}[\hat\delta]\ ,\\
R^{(g)}_k\thspin{\hat\delta'&1}{\hat\delta''&0}=&\sum_{\lambda_k\in\Lambda_e+\frac{u}{2}}R^{(g)}_{k-1}[\hat\delta]
+\sum_{\lambda_k\in\Lambda_o+\frac{u}{2}}R^{(g)}_{k-1}[\hat\delta]\
,
\\R^{(g)}_k\thspin{\hat\delta'&1}{\hat\delta''&1}=&\sum_{\lambda_k\in\Lambda_e+\frac{u}{2}}R^{(g)}_{k-1}[\hat\delta]
-\sum_{\lambda_k\in\Lambda_o+\frac{u}{2}}R^{(g)}_{k-1}[\hat\delta]\
,
\end{align*}
so that $$
R^{(g)}_k\thspin{\hat\delta'&0}{\hat\delta''&0}+R^{(g)}_k\thspin{\hat\delta'&0}{\hat\delta''&1}
+R^{(g)}_k\thspin{\hat\delta'&1}{\hat\delta''&0}=\sum_{\lambda_k\in\Lambda^{(1)}}R^{(g)}_{k-1}[\hat\delta]
+\sum_{\lambda_k\in\Lambda^{(2)}}R^{(g)}_{k-1}[\hat\delta]\ ,$$ and
$$ R^{(g)}_k\thspin{\hat\delta'&1}{\hat\delta''&1}=\sum_{\lambda_k\in\Lambda^{(1)}}R^{(g)}_{k-1}[\hat\delta]
-\sum_{\lambda_k\in\Lambda^{(2)}}R^{(g)}_{k-1}[\hat\delta]\ .$$
{}From these formulas, eq.\eqref{tobeproved} for $k=1$ follows
immediately. Now, suppose that eq.\eqref{tobeproved} holds for
$k-1$. Then
\begin{equation*}\begin{split} \sum_{\delta\in\FF_2^{2k}\text{ even}}&R^{(g)}_k[\delta]=\sum_{\hat\delta\in\FF_2^{2(k-1)}\text{
even}}\Bigl(R^{(g)}_k\thspin{\hat\delta'&0}{\hat\delta''&0}+R^{(g)}_k\thspin{\hat\delta'&0}{\hat\delta''&1}
+R^{(g)}_k\thspin{\hat\delta'&1}{\hat\delta''&0}\Bigr)+\sum_{\hat\delta\in\FF_2^{2(k-1)}\text{
odd}}R^{(g)}_k\thspin{\hat\delta'&1}{\hat\delta''&1}\\
=&2^{k-2}\Biggl(\sum_{\substack{\lambda_k\in\Lambda^{(1)}\\
\lambda_1,\ldots,\lambda_{k-1}\in\Lambda^{(1)}}}\!\!\!\!\!\!\!R^{(g)}_0+
\sum_{\substack{\lambda_k\in\Lambda^{(1)}\\
\lambda_1,\ldots,\lambda_{k-1}\in\Lambda^{(2)}}}\!\!\!\!\!\!\!R^{(g)}_0+
\sum_{\substack{\lambda_k\in\Lambda^{(2)}\\
\lambda_1,\ldots,\lambda_{k-1}\in\Lambda^{(1)}}}\!\!\!\!\!\!\!R^{(g)}_0+
\sum_{\substack{\lambda_k\in\Lambda^{(2)}\\
\lambda_1,\ldots,\lambda_{k-1}\in\Lambda^{(2)}}}\!\!\!\!\!\!\!R^{(g)}_0\\
&+\sum_{\substack{\lambda_k\in\Lambda^{(1)}\\
\lambda_1,\ldots,\lambda_{k-1}\in\Lambda^{(1)}}}\!\!\!\!\!\!\!R^{(g)}_0-
\sum_{\substack{\lambda_k\in\Lambda^{(1)}\\
\lambda_1,\ldots,\lambda_{k-1}\in\Lambda^{(2)}}}\!\!\!\!\!\!\!R^{(g)}_0-
\sum_{\substack{\lambda_k\in\Lambda^{(2)}\\
\lambda_1,\ldots,\lambda_{k-1}\in\Lambda^{(1)}}}\!\!\!\!\!\!\!R^{(g)}_0+
\sum_{\substack{\lambda_k\in\Lambda^{(2)}\\
\lambda_1,\ldots,\lambda_{k-1}\in\Lambda^{(2)}}}\!\!\!\!\!\!\!R^{(g)}_0\Biggr)\\
=&2^{k-1}\Bigl(\sum_{\lambda_1,\ldots,\lambda_{k}\in\Lambda^{(1)}}\!\!\!\!\!\!\!R^{(g)}_0
+\sum_{\lambda_1,\ldots,\lambda_{k}\in\Lambda^{(2)}}\!\!\!\!\!\!\!R^{(g)}_0\Bigr)\ .
\end{split}\end{equation*} An analogous computation gives the case with odd spin
structures. \end{proof}

\begin{corollary}
For the lattices $\Lambda_k$, $k=1,\ldots,5$ in table
\ref{t:lattices}, $\Lambda^{(1)}\cong\Lambda^{(2)}$.
\end{corollary}
\begin{proof}
Notice that for these lattices
$$
\Theta_k[\delta]=\theta[\delta]^{n_k}\Theta_{\tilde\Lambda_k}[\delta]\
, $$ with $n_k>0$. It follows that $\Theta_k[\delta]=0$ if $\delta$
is odd. By proposition \ref{p:seriessum}, this implies
$$ \Theta^{(g)}_{\Lambda^{(1)}}=\Theta^{(g)}_{\Lambda^{(2)}}\ ,
$$ for all $g$ and, since $E_8^2$ and $D_{16}^+$ have
different theta series at $g=4$, one gets
$\Lambda^{(1)}\cong\Lambda^{(2)}$.
\end{proof}

As an application, we can use this result to compute the constant
$C$ in \eqref{Ccost}. By summing both sides of \eqref{reltheta4} over all even
spin structures, we obtain
$$
2^3(c_0^5+2c^5_2+2c^5_4)\Theta_{E_8^2}+2^3(c_0^5+2c^5_1+2c^5_3+2c^5_5)\Theta_{D_{16}^+}=2^3(2^4+1)C(\Theta_{E_8^2}-\Theta_{D_{16}^+})\
, $$ so that
$$ C=\frac{c_0^5+2c^5_2+2c^5_4}{17}=-\frac{c_0^5+2c^5_1+2c^5_3+2c^5_5}{17}=-\frac{2^5\cdot 3}{7}\ .
$$
An analogous calculation gives the constants $B_4$ and $B_5$.

\end{document}